\documentclass[11pt,a4paper]{article}

\usepackage{amsmath,amsfonts,amssymb,amsthm}
\usepackage{mathbbol}
\usepackage{amsthm}
\usepackage{graphicx,color}
\usepackage{boxedminipage}
\usepackage[linesnumbered, boxed, algosection,noline]{algorithm2e}
\usepackage{framed}
\usepackage{thmtools}
\usepackage{thm-restate}
\usepackage{xspace}
\usepackage{vmargin}
\setmarginsrb{1in}{1in}{1in}{1in}{0pt}{0pt}{0pt}{6mm}
\usepackage{todonotes}
 \usepackage[pdftex, plainpages = false, pdfpagelabels, 
                 bookmarks=false,
                 bookmarksopen = true,
                 bookmarksnumbered = true,
                 breaklinks = true,
                 linktocpage,
                 pagebackref,
                 colorlinks = true,  
                 linkcolor = blue,
                 urlcolor  = blue,
                 citecolor = red,
                 anchorcolor = green,
                 hyperindex = true,
                 hyperfigures
                 ]{hyperref} 
 \usepackage{xifthen}
 \usepackage{tabularx}
\usepackage{epigraph}
 \usetikzlibrary{calc}

\DeclareMathOperator{\operatorClassNP}{{\sf NP}}
\newcommand{\classNP}{\ensuremath{\operatorClassNP}\xspace}
\DeclareMathOperator{\operatorClassCoNP}{{\sf co-NP}}
\newcommand{\classCoNP}{\ensuremath{\operatorClassCoNP}\xspace}
\DeclareMathOperator{\operatorClassFPT}{{\sf FPT}\xspace}
\newcommand{\classFPT}{\ensuremath{\operatorClassFPT}\xspace}
\DeclareMathOperator{\operatorClassW}{{\sf W}}
\newcommand{\classW}[1]{\ensuremath{\operatorClassW[#1]}\xspace}
\DeclareMathOperator{\operatorClassCoW}{{\sf co-W}}
\newcommand{\classCoW}[1]{\ensuremath{\operatorClassCoW[#1]}\xspace}

\DeclareMathOperator{\operatorClassXP}{{\sf X}P\xspace}
\newcommand{\classXP}{\ensuremath{\operatorClassXP}\xspace}
 \DeclareMathOperator{\operatorClassPSPACE}{{\sf PSPACE}\xspace}
\newcommand{\classPSPACE}{\ensuremath{\operatorClassPSPACE}\xspace}
 \DeclareMathOperator{\operatorClassETH}{{\sf ETH}\xspace}
\newcommand{\classETH}{\ensuremath{\operatorClassETH}\xspace}

\newcommand{\CopsR}{\ProblemName{Cops-and-Robber}}
\newcommand{\divider}{Divider\xspace}
\newcommand{\dinumber}{dynamic separation\xspace}

\newcommand{\Oh}{\mathcal{O}}

\newtheorem{theorem}{Theorem}

\newtheorem{claim}{Claim}
\newtheorem{corollary}{Corollary}

\newtheorem{observation}{Observation}
\newtheorem{proposition}{Proposition}

\newcommand{\true}{{\sf true}\xspace}
\newcommand{\false}{{\sf false}\xspace}
\newcommand{\dist}{{\sf dist}}
\newcommand{\nd}{{\sf nd}}
\newcommand{\id}{{\sf id}}

\newcommand{\pname}{\textsc}
\newcommand{\ProblemFormat}[1]{\pname{#1}}
\newcommand{\ProblemIndex}[1]{\index{problem!\ProblemFormat{#1}}}
\newcommand{\ProblemName}[1]{\ProblemFormat{#1}\ProblemIndex{#1}{}\xspace}

 \newcommand{\probR}{\ProblemName{Rendezvous}}
 \newcommand{\probRT}{\ProblemName{Rendezvous in Time}}
 \newcommand{\probILP}{\ProblemName{Integer Linear Programming Feasibility}}



\newlength{\RoundedBoxWidth}
\newsavebox{\GrayRoundedBox}
\newenvironment{GrayBox}[1]%
   {\setlength{\RoundedBoxWidth}{.93\textwidth}
    \def\boxheading{#1}
    \begin{lrbox}{\GrayRoundedBox}
       \begin{minipage}{\RoundedBoxWidth}}%
   {   \end{minipage}
    \end{lrbox}
    \begin{center}
    \begin{tikzpicture}%
       \node(Text)[draw=black!20,fill=white,rounded corners,%
             inner sep=2ex,text width=\RoundedBoxWidth]%
             {\usebox{\GrayRoundedBox}};
        \coordinate(x) at (current bounding box.north west);
        \node [draw=white,rectangle,inner sep=3pt,anchor=north west,fill=white] 
        at ($(x)+(6pt,.75em)$) {\boxheading};
    \end{tikzpicture}
    \end{center}}     

\newenvironment{defproblemx}[2][]{\noindent\ignorespaces%
                                \FrameSep=6pt%
                                \parindent=0pt%
                \vspace*{-1.5em}
                \ifthenelse{\isempty{#1}}{%
                  \begin{GrayBox}{\textsc{#2}}%
                }{%
                  \begin{GrayBox}{\textsc{#2}  parameterized by~{#1}}%
                }
                \begin{tabular*}{\textwidth}{@{\hspace{.1em}} >{\itshape} p{1.8cm} p{0.8\textwidth} @{}}%
            }{
                \end{tabular*}%
                \end{GrayBox}%
                \ignorespacesafterend
            }

\newcommand{\defproblema}[3]{
  \begin{defproblemx}{#1}
    Input:  & #2 \\
    Task: & #3
  \end{defproblemx}
}%

\begin{document}
\title{Can Romeo  and  Juliet Meet?\\
Or Rendezvous Games with Adversaries on Graphs\thanks{This research was supported by  the French Ministry of Europe and Foreign Affairs, via the Franco-Norwegian project PHC AURORA. The two first authors have been supported by the Research Council of Norway via the project  ``MULTIVAL" (grant no. 263317). The last author was supported by the ANR projects DEMOGRAPH (ANR-16-CE40-0028), ESIGMA (ANR-17-CE23-0010), ELIT (ANR-20-CE48-0008), and the French-German Collaboration ANR/DFG Project UTMA (ANR-20-CE92-0027).}}

\author{
Fedor V. Fomin\thanks{
Department of Informatics, University of Bergen, Norway.} \addtocounter{footnote}{-1}
\and
Petr A. Golovach\footnotemark{} 
\and
Dimitrios M. Thilikos\thanks{LIRMM, Univ Montpellier, CNRS, Montpellier, France.}
}

\date{}

\maketitle

\epigraph{\scriptsize For never was a story of more woe than this of Juliet and {her} Romeo.}{ ---{\sl\scriptsize  William Shakespeare, Romeo and Juliet}}
\begin{abstract}
\noindent We introduce the rendezvous game with adversaries. In this game, two players, {\sl Facilitator} and {\sl \divider}, play against each other on a graph. Facilitator has two agents, and \divider has a team of $k$ agents located in some vertices of the graph. They take turns in moving their agents to adjacent vertices (or staying put). Facilitator wins if his agents meet in some vertex of the graph. The goal of \divider is to prevent the rendezvous of Facilitator's agents. Our interest is to decide whether Facilitator can win.
It appears that, in general, the problem is \classPSPACE-hard and, when parameterized by $k$,  \classCoW{2}-hard. Moreover, even the game's variant where we ask whether 
Facilitator can ensure the meeting of his agents
within $\tau$ steps is  \classCoNP-complete already for $\tau=2$. On the other hand, for chordal and $P_5$-free graphs, we prove that the problem is solvable in polynomial time. These algorithms exploit an interesting relation of the game and minimum vertex cuts in certain graph classes. Finally, we show that the problem is fixed-parameter tractable parameterized by both the graph's neighborhood diversity and $\tau$.
\end{abstract}

\section{Introduction}\label{sec:intro}

We introduce the \emph{Rendezvous Game with Adversaries} on graphs. In our game, a team of dividers tries to prevent two passionate lovers, say Romeo and Juliet, from meeting each other.  We are interested in the minimum size of the team of dividers sufficient to obstruct their romantic encounter.
 In the static setting, when dividers do not move, this is just the problem of computing the minimum vertex cut between the pair of vertices occupied by Romeo and Juliet. But in the dynamic variant, when dividers are allowed to change their position, the team's size can be significantly smaller than the size of the minimum cut. In fact, this gives rise to a 
new interactive form of connectivity that is much more challenging 
both from the combinatorial and the algorithmic point of view.

Our rendezvous game rules are very similar to the rules of the classical \CopsR game of Nowakowski-Winkler and Quillioit \cite{NowakowskiW83,Quilliot85}, see also the book of Bonato and Nowakowski \cite{BonatoN11}. The difference is that in the \CopsR  game, a team of $k$ cops tries to capture a robber in a graph, while in our game, the group of $k$ dividers tries to keep the two   lovers separated. 

A bit more formally. 
The game is played on a finite undirected connected graph $G$ by two players: \emph{Facilitator} and \emph{\divider}. 
Facilitator has two agents $R$ and $J$ that are initially placed in designated vertices $s$ and $t$ of $G$. \divider has a team of $k\geq 1$ agents $D_1,\ldots,D_k$ that are initially placed in some vertices of $V(G)\setminus\{s,t\}$ selected by \divider. Several \divider's agents can occupy the same vertex. Then the players make their moves by turn, starting with Facilitator. At every move, each player moves some of {his} agents to adjacent vertices or keeps them in their old positions. No agent can be moved to a vertex that is currently occupied by adversary agents. Both players have complete information about $G$ and the positions of all the agents. Facilitator aims to ensure that $R$ and $J$ meet; that is, they are in the same vertex. The task of \divider is to prevent the rendezvous of $R$ and $J$ by maintaining $D_1,\ldots,D_k$ in positions that block the possibility to meet. Facilitator wins if $R$ and $J$ meet, and \divider wins if {he} succeed in preventing the meeting of $R$ and $J$ forever.

We define the following problem:

\defproblema{\probR}{A graph $G$ with two given vertices $s$ and $t$, and a positive integer $k$. }{Decide whether Facilitator can win on $G$ starting from $s$ and $t$ against \divider with $k$ agents.}
 
Another variant of the game is when the number of moves of the players is at most some parameter $\tau$. Then Facilitator wins if $R$ and $J$ meet within the first $\tau$ moves, and \divider wins otherwise. Thus the problem is.

 \defproblema{\probRT}{A graph $G$ with two given vertices $s$ and $t$, and positive integers $k$ and $\tau$. }{Decide whether Facilitator can win on $G$ starting from $s$ and $t$ in at most $\tau$ steps against \divider with $k$ agents.}
 Notice that, in the above problem, $\tau$ is part of the input. We also consider the version of the problem where $\tau$ is a fixed constant. This generates a family of problems, one  for each different value of $\tau$, and we refer to each of them as the $\tau$-\probRT problem.

\medskip
\noindent{\bf Our results.} We start with combinatorial results. 
If $s=t$ or if $s$ and $t$ are adjacent, then Facilitator wins by a trivial strategy. However, if $s$ and $t$ are distinct nonadjacent vertices, then \divider wins provided that {he} has sufficiently many agents. For example, the agents can be placed in the vertices of an $(s,t)$-separator and stay there.  Then $R$ and $J$ never meet. 
We call the minimum number $k$ of the agents of \divider that is sufficient for {his} winning, the \emph{$(s,t)$-\dinumber} number and denote it by $d_G(s,t)$. We put  $d_G(s,t)=+\infty$ for $s=t$ or when $s$ and $t$ are adjacent. The \dinumber number can be seen as a dynamic analog of the minimum size $\lambda_G(s,t)$ of a vertex  $(s,t)$-separator in $G$. (The minimum number of vertices whose removal leaves $s$ and $t$ in different connected components.)  
Then \probR can be restated as the problem of  deciding whether $d_G(s,t)>k$. 

The first natural question is: What is the relation between $d_G(s,t)$ and $\lambda_G(s,t)$? Clearly, $d_G(s,t)\leq \lambda_G(s,t)$. We show 
 that $d_G(s,t)=1$ if and only if $\lambda_G(s,t)=1$. 
If $d_G(s,t)\geq 2$, then we construct examples demonstrating that  the difference $\lambda_G(s,t)-d_G(s,t)$ can be arbitrary even for sparse graphs. Interestingly, there are  graph classes where both parameters are equal.  In particular, we show that $\lambda_G(s,t)=d_G(s,t)$ holds for $P_5$-free graphs and chordal graphs. This also yields a polynomial time algorithm computing $d_G(s,t)$ on these classes of graphs. 

Then we turn to the computational complexity of \probR and \probRT on general graphs. Both problems can be solved it $n^{\Oh(k)}$ time by using a backtracking technique. 
We show that this running time is asymptotically tight by proving that they are both \classCoW{2}-hard when parameterized by $k$ (we prove that it is \classW{2}-hard to decide whether $d_G(s,t)\leq k$) and cannot be solved in  $f(k) \cdot  n^{o(k)}$ time for any function $f$ of $k$, unless $\classETH$ fails. Moreover,  $\tau$-\probRT is \classW{2}-hard, for every   $\tau\geq 2$. If $\tau$ is a constant, then $\tau$-\probRT is in \classCoNP and our \classCoW{2}-hardness proof implies that $\tau$-\probRT is \classCoNP-complete for every   $\tau\geq 2$. For the general case, the problems are even harder as we prove that  \probR and \probRT are both \classPSPACE-hard.

Finally, we initiate  the study of the complexity of the problems under structural parameterization of the input graphs. We show that \probRT is \classFPT when parameterized by the neighborhood diversity of the input graph and $\tau$.

\medskip
\noindent{\bf Related work.}
The classical {rendezvous} game introduced by Alpern~\cite{Alpern95} is played by two agents that are placed in some unfamiliar area and whose task is to develop strategies that  maximize the probability that they meet. We refer to the book of Alpern  and Gal \cite{AlpernG03-Th} for detailed study of the subject. Deterministic rendezvous problem was studied by Ta-Shma and Zwock \cite{DBLP:journals/talg/Ta-ShmaZ14}. See also  \cite{DessmarkFKP06,FraigniaudP13} for other variants of rendezvous
problems on graphs. 

\probR is closely related to the  \CopsR game.   The game was defined (for one cop) by Winkler and
Nowakowski \cite{NowakowskiW83} and Quilliot \cite{Quilliot85} who
also characterized graphs for which one cop can catch the robber.
  Aigner and Fromme \cite{AignerF84} initiated the study of the problem with
several cops. The minimum number of cops that are required to capture the 
robber is called the cop number of a graph. This problem was studied 
intensively and we
 refer to   the book of Bonato and Nowakowski \cite{BonatoN11} for further  references. 
Kinnersley  \cite{Kinnersley15} established   that the problem is {\sf EXPTIME}-complete.
 The 
\CopsR game can be seen as a special case of search games played 
on graphs,   surveys \cite{bonato2013graph,FominT08} provide further
references on search and pursuit-evasion games on graphs.  
A   related variant of \CopsR game is the  
guarding game studied in 
\cite{DBLP:journals/algorithmica/FominGHMVW11,FominGL11,nagamochi2011cop,SamalV14}. Here the set of cops is trying to prevent the robber from entering a specified subgraph in a graph.

\medskip
\noindent{\bf Organization of the paper.} In Section~\ref{sec:prelim}, we give the basic definitions and introduce notation used throughout the paper. We also show that \probR and \probRT can be solved in $n^{\Oh(k)}$ time. In Section~\ref{sec:sep}, 
we investigate relations between $d_G(s,t)$ and $\lambda_G(s,t)$. In Section~\ref{sec:hardness}, we give algorithmic lower bounds for \probR and \probRT. 
In Section~\ref{sec:diversity}, we show that \probRT is fixed-parameter tractable  (\classFPT) when parameterized by $\tau$ and the neighborhood  diversity of the input graph.
We conclude in Section~\ref{sec:concl} by stating some open problems.

\section{Preliminaries}\label{sec:prelim}  
\noindent
{\bf Graphs.}
All graphs considered in this paper are finite undirected graphs without loops or multiple edges, unless it is said explicitly that we consider directed graphs. 
We follow the standard graph theoretic notation and terminology (see, e.g., \cite{Diestel12}). For each of the graph problems considered in this paper, we let $n=|V(G)|$ and $m=|E(G)|$ denote the number of vertices and edges,
respectively, of the input graph $G$ if it does not create confusion. 
For a graph $G$ and a subset $X\subseteq V(G)$ of vertices, we write $G[X]$ to denote the subgraph of $G$ induced by $X$.
For a set of vertices $S$, $G-S$ denotes the graph obtained by deleting the vertices of $S$, that is, $G-S=G[V(G)\setminus S]$; for a vertex $v$, we write $G-v$ instead of $G-\{v\}$.
For a vertex $v$, we denote by $N_G(v)$ the \emph{(open) neighborhood} of $v$, i.e., the set of vertices that are adjacent to $v$ in $G$.
We use $N_G[v]$ to denote the \emph{closed neighborhood}, that is $N_G(v)\cup \{v\}$.
For two nonadjacent vertices $s$ and $t$, a set of vertices $S\subseteq V(G)\setminus\{s,t\}$ is an \emph{$(s,t)$-separator} if $s$ and $t$ are in distinct connected components of $G-S$. We use $\lambda_G(s,t)$ to denote the minimum size of an $(s,t)$-separator of $G$; $\lambda_G(s,t)=+\infty$ if $s=t$ or $s$ and $t$ are adjacent. 
A \emph{path} is a connected graph with at leat one and most two vertices (called \emph{end-vertices}) of degree at most one whose remaining vertices (called \emph{internal}) have degrees two. We say that a path with end-vertices $u$ and $v$ is an \emph{(u,v)-path}. The \emph{length} of a path $P$, denoted by $\ell(P),$ is the number of its edges.
The \emph{distance} $\dist_G(u,v)$ between two vertices $u$ and $v$ of $G$ in the length of a shortest $(u,v)$-path.
We use $v_1\cdots v_k$ to denote the path with the vertices $v_1,\ldots,v_k$ and the edges $v_{i-1}v_i$ for $i\in\{2,\ldots,k\}$.
A \emph{cycle} is a connected graph with all the vertices of degree two. The \emph{length} $\ell(C)$ of a cycle $C$ is the number of edges of $C$.

Let $X$ and $Y$ be multisets of vertices of a graph $G$  (i.e., $X$ and $Y$ can contain several copies of the same vertex). We say that $X$ and $Y$ of the same size are \emph{adjacent} if there is a bijective mapping $\alpha\colon X\rightarrow Y$ such that for $x\in X$, either $x=\alpha(x)$ or $x$ and $\alpha(x)$ are adjacent in $G$. It is useful to observe the following.

\begin{observation}\label{obs:adj}
For multisets $X$ and $Y$ of vertices of $G$, it can be decided in polynomial time whether $X$ and $Y$ are adjacent. 
\end{observation}

\begin{proof}
It is trivial to check whether $X$ and $Y$ have the same size. If this holds, we construct the bipartite graph $H$ with the vertex set $V_1\cup V_2$, where $|V_1|=|V_2|=|X|=|Y|$, and the nodes of $V_1$ correspond to the elements of $X$ and the nodes of $V_2$ correspond to the elements of $Y$. A node of $V_1$ is adjacent to a node of $V_2$ if and only if the corresponding vertices of $G$ are either the same or adjacent. Then $X$ and $Y$ are adjacent if and only if $H$ has a perfect matching. The existence of a perfect matching in a bipartite graph can be verified in polynomial time (see, e.g., \cite{Lovasz09}) and the claim follows. 
\end{proof}

\medskip
\noindent
{\bf Parameterized Complexity.} We obtain a number of results about the parameterized complexity of \probR and \probRT. We refer to the resent book of Cygan et al.~\cite{CyganFKLMPPS15} for the introduction to the area. Here we just remind that an instanse of the parameterized version $\Pi_p$ of a decision problem $\Pi$ is a pair $(I,k)$, where $I$ is an instance of $\Pi$ and $k$ is an integer \emph{parameter} associated with $I$.  It is said that $\Pi_p$ is \emph{fixed-parameter tractable} (\classFPT) if it can be solved in time $f(k)|I|^{\Oh(1)}$ for a computable function $f(k)$ of the parameter $k$. The Parameterized Complexity theory also provides tools that allow to show that a parameterized problem cannot be solved in \classFPT time (up to some reasonable complexity assumptions). 
For this, Downey and Fellows (see~\cite{DowneyF13}) introduced a hierarchy of parameterized complexity classes, namely  
$\classFPT\subseteq \classW{1}\subseteq \classW{2}\subseteq\cdots\subseteq \classXP$, and the basic conjecture is that all inclusions in the hierarchy are proper. The usual way to show that it is unlikely that a parameterized problem admit an \classFPT algorithm is to show that it is \classW{1} or \classW{2}-hard using a \emph{parameterized reduction} from a known hard problem in the corresponding class. The most common tool for establishing fine-grained complexity lower bound for parameterized problems is the \emph{Exponential Time Hypothesis} (\classETH) proposed by Impagliazzo, Paturi, and Zane~\cite{ImpagliazzoP01,ImpagliazzoPZ01}. 
This is the conjecture stating that there is $\varepsilon>0$ such that \textsc{$3$-Satisfiability} cannot be solved in $\Oh^*(2^{\varepsilon n})$ time on formulas with $n$ variables.  

\medskip
\noindent
{\bf Rendezvous Games with Adversaries.} 
Suppose that the game is played on a connected graph $G$, and $s$ and $t$ are initial positions of the agents of Facilitator. Let also $k$ be the number of agents of \divider. 

Notice that a placements of the agents of Facilitator is defined by a multiset of two vertices, as $R$ and $J$ can occupy the same vertex. 
We denote by $\mathcal{F}_G$ the family of all multisets of two vertices. 
Similarly, a placement of  $k$ agents of \divider is defined by a multiset of $k$ vertices, because several agents can occupy the same vertex. Let $\mathcal{D}_G^k$ be the family of all multisets of $k$ vertices. We say that $F\in\mathcal{F}_G$ and $D\in\mathcal{D}_G^k$ are \emph{compatible} if $F\cap D=\emptyset$. Notice that the number of pairs of compatible  $F\in\mathcal{F}_G$ and $D\in \mathcal{D}_G^k$ is $n\binom{n+k-2}{k}+\binom{n}{2}\binom{n+k-3}{k}$. We denote by 
$$\mathcal{P}_G^k=\{(F,D)\mid F\in\mathcal{F}_G,~D\in \mathcal{D}_G^k\text{ s.t. }F\text{ and }D\text{ are compatible} \}$$
the set of \emph{positions} in the game.

Formally, a \emph{strategy} of Facilitator for \probR is a function 
$f\colon \mathcal{P}_G^k\rightarrow \mathcal{F}_G$ that maps $(F,D)\in \mathcal{P}_G^k$ to $F'\in \mathcal{F}_G$ such that $F$ and $F'$ are adjacent and $F'$ is compatible with $D$. In words, given a position $(F,D)$, Facilitator moves $R$ and $J$ from $F$ to $F'$ if this is his turn to move. Similarly, a \emph{strategy} of \divider is a function 
$d\colon \mathcal{P}_G^k\rightarrow \mathcal{D}_G^k$ that maps $(F,D)\in \mathcal{P}_G^k$ to $D'\in \mathcal{F}_G^k$ such that $D$ and $D'$ are adjacent and $D'$ is compatible with $F$, that is, \divider moves his agents from $D$ to $D'$ if this is his turn to move. To accommodate the initial placement, we extend the definition of $d$ for the pair $(\{s,t\},\emptyset)$ and let $d((\{s,t\},\emptyset)=D'$, where $D'\in\mathcal{D}_G^k$ is compatible with $\{s,t\}$.

The definitions of strategies for \probRT are more complicated, because the decisions of the players also depend on the number of the current step. 
A \emph{strategy} of Facilitator for \probR is a family of functions  $f_i\colon \mathcal{P}_G^k\rightarrow \mathcal{F}_G$ for $i\in\{1,\ldots,\tau\}$ such that
$f_i$ maps $(F,D)\in \mathcal{P}_G^k$ to $F'\in \mathcal{F}_G$, where  $F$ and $F'$ are adjacent and $F'$ is compatible with $D$. Facilitator uses $f_i$ for the move in the $i$-th step of the game. A  \emph{strategy} of \divider is a family of functions  
$d_i\colon \mathcal{P}_G^k\rightarrow \mathcal{D}_G^k$ for $i\in\{0,\ldots,\tau-1\}$ such that for $i\in\{1,\ldots,\tau-1\}$,
$d_i$ maps $(F,D)\in \mathcal{P}_G^k$ to $D'\in \mathcal{F}_G^k$, where  $D$ and $D'$ are adjacent and $D'$ is compatible with $F$, and $d_0$ maps $(\{s,t\},\emptyset)$ to $D'\in\mathcal{D}_G^k$ compatible with $\{s,t\}$ (slightly abusing notation we do not define $d_0$ for the elements of $\mathcal{P}_G^k$).   

In the majority of the proofs in our paper, we rather explain the strategies of the players in an informal way, to avoid defining functions for all elements of $\mathcal{P}_G^k$, because the majority of positions never occur in the game. However, the above notation is useful in some cases. 

As it is common for various games on graphs (see, e.g., the book of Bonato and Nowakowski~\cite{BonatoN11} about \CopsR games), our Rendezvous Game with Adversaries can be resolved by backtracking. As the approach is standard, we only briefly sketch the proof of the following theorem.

\begin{theorem}\label{thm:backtrack}
\probR and \probRT can be solved in $n^{\Oh(k)}$ time.
\end{theorem}

\begin{proof}
Let $G$ be a connected graph on which the game is played and let $s,t\in V(G)$. Let also $k$ be a positive integer denoting the number of agents of \divider. 
 
We define the \emph{game} graph $\mathcal{G}$ (also the name \emph{arena} could be found in the literature) as the directed graph, whose nodes correspond to positions and turns to move.  We denote the nodes of $\mathcal{G}$ by $v_{F,D}^{(h)}$ for $(F,D)\in\mathcal{P}_G^k$, and $h\in\{1,2\}$; if $h=1$, then Facilitator makes a move, and if $h=2$, then this is \divider's turn. 
For $h\in\{1,2\}$, we set $\mathcal{V}_h=\{v_{F,D}^{(h)}\mid (F,D)\in\mathcal{P}_G^k \}$.
For every two nodes $v_{F,D}^{(1)}\in \mathcal{V}_1$ and $v_{F',D'}^{(2)}\in\mathcal{V}_2$, we construct arcs as follows. We construct the arc $(v_{F,D}^{(1)},v_{F',D'}^{(2)})$ if $D=D'$, and  $F$ and $F'$ are adjacent. Symmetrically, we construct  $(v_{F',D'}^{(2)},v_{F,D}^{(1)})$ if $F=F'$, and  $D'$ and $D$ are adjacent. We denote by $\mathcal{A}$ the set of arcs of $\mathcal{G}$.   

Observe, that $\mathcal{G}$ can be constructed in $n^{\Oh(k)}$ time. The number of nodes is $2\cdot(n\binom{n+k-2}{k}+\binom{n}{2}\binom{n+k-3}{k})= n^{\Oh(1)}$. 
Given two nodes $v_{F,D}^{(1)}$ and $v_{F',D'}^{(2)}$, the arcs between these nodes can be constructed in polynomial time by Observation~\ref{obs:adj}. Hence, the construction of the arc set can be done in time $k^{\Oh(k)}\cdot n^{\Oh(k)}\cdot n^{\Oh(k)}=n^{\Oh(k)}$.

Let $\ell\geq 0$ be an integer.  We define the set $\mathcal{W}_\ell\subseteq \mathcal{V}_1$ of \emph{winning positions} for Facilitator in at most $\ell$ moves.
A node $v_{F,D}^{(1)}$ is in $\mathcal{W}_\ell$ if Facilitator can win on $G$ in at most $\ell$ moves provided that $R$ and $J$ are placed in $F$ and the agents of \divider are occupying $D$. We explain how to construct $\mathcal{W}_\ell$ for $\ell=0,1,\ldots$ by dynamic programming.  

It is straightforward to verify that $v_{F,D}^{(1)}\in \mathcal{W}_0$ if and only if $F=\{x,x\}$ for $x\in V(G)$.  

For $\ell\geq 1$,
\begin{equation}\label{eq:backtrack} 
\mathcal{W}_\ell=\mathcal{W}_{\ell-1}\cup \mathcal{U},
\end{equation}
where
\begin{equation}\label{eq:backtrack-two}
\mathcal{U}=\{v_{F,D}^{(1)}\mid\text{ there is }(v_{F,D}^{(1)},v_{F',D}^{(2)})\in\mathcal{A}\text{ s.t. for every }(v_{F',D}^{(2)},v_{F',D'}^{(1)})\in\mathcal{A},~v_{F',D'}^{(1)}\in \mathcal{W}_{\ell-1}\}.
\end{equation}
Informally, $v_{F,D}^{(1)}\in \mathcal{U}$ if there is a move of Facilitator such that for every response of \divider, the obtained position is in $\mathcal{W}_{\ell-1}$, that is, Facilitator can win in at most $\ell-1$ steps from this position.

The correctness of computing $\mathcal{W}_\ell$ using (\ref{eq:backtrack}) and (\ref{eq:backtrack-two}) is proved by completely standard arguments and we leave this to the reader. Notice that given $\mathcal{W}_{\ell-1}$, (\ref{eq:backtrack}) and (\ref{eq:backtrack-two}) allow to compute $\mathcal{W}_\ell$ in $n^{\Oh(1)}$ time.

We compute the sets $\mathcal{W}_\ell$ consecutively starting from $\ell=0$ until we obtain $W_{\ell}=W_{\ell-1}$ for some $\ell\geq 1$. Observe that if $W_{\ell}=W_{\ell-1}$, then $W_{\ell'}=W_{\ell}$ for every $\ell'\geq \ell$. Notice also that we stop after at most $|\mathcal{V}_1|$ iterations, because $\mathcal{W}_\ell\subseteq \mathcal{V}_1$. This implies that all the sets $\mathcal{W}_\ell$ can be constructed in  $n^{\Oh(1)}$ time. Let $\mathcal{W}_{\ell^*}$ be the last constructed set.

To solve \probR for an instance $(G,s,t,k)$, it is sufficient to observe that $(G,s,t,k)$ is a yes-instance if and only if $v_{F,D}^{(1)}\in \mathcal{W}_{\ell^*}$ for $F=\{s,t\}$ and every $D\in\mathcal{D}_G^k$ that is compatible with $F$, that is, Facilitator can win starting from $s$ and $t$ for every choice the initial placement of the $k$ agents of \divider.

Solving \probRT is slightly more complicated, because the parameter $\tau$ is expected to be encoded in binary. Let $(G,s,t,k,\tau)$ be an instance of \probRT. If $\tau\geq \ell^*$, we observe that $(G,s,t,k,\tau)$ is a yes-instance of \probRT if and only if $(G,s,t,k)$ is a yes-instance of \probR. If $\tau<\ell^*$, then recall that we already constructed the set $\mathcal{W}_\tau$. Then $(G,s,t,k,\tau)$ is a yes-instance if and only if  $v_{F,D}^{(1)}\in \mathcal{W}_{\tau}$ for $F=\{s,t\}$ and every $D\in\mathcal{D}_G^k$ that is compatible with $F$.

Summarizing the running time of all the steps, we obtain that \probR and \probRT can be solved in $n^{\Oh(k)}$ time.
\end{proof}

We conclude this section by the observation that a strategy of \divider in Rendezvous Games with Adversaries for $\tau$ steps, that is, for \probRT, can be represented as a rooted tree of height $\tau$. Suppose that the game is played on a graph $G$, and $s$ and $t$ are initial positions of the agents of Facilitator. Let also $k$ be the number of agents of \divider. Suppose that \divider has a strategy defined by the family of functions $d_i\colon \mathcal{P}_G^k\rightarrow \mathcal{D}_G^k$ for $i\in\{0,\ldots,\tau-1\}$.
We define the tree $\mathcal{T}_G^k(\tau)$ such that every node $v$ of $\mathcal{T}_G^k(\tau)$ is associated with a position 
$P_v\in \mathcal{P}_G^k$
inductively starting from the root:
\begin{itemize}
\item $P_r=(\{s,t\},d_0(\{s,t\},\emptyset))$ is associated with the root $r$ of $\mathcal{T}_G^k(\tau)$.
\item for every node $v\in V(\mathcal{T}_G^k(\tau))$ with $P_v=(F,D)$ 
at distance $i\leq\tau-1$ from the root, 
we construct a child $u$ of $v$ for every $(F',D')\in \mathcal{P}_G^k$ such that 
(i) $F'$ is adjacent to $F$ and compatible with $D$, and (ii) $D'=d_i(F',D)$, and associate $u$ with  $P_u=(F',D')$. 
\end{itemize}
In words, the children of every node correspond to all possible moves of Facilitator from the position $(D,F)$ and are the positions obtained by the responses of \divider. 
Observe that every node has at most $|\mathcal{F}|=\binom{n+1}{2}$ children.  Therefore, the total number of nodes of $\mathcal{T}_G^k(\tau)$ is at most $\binom{n+1}{2}^{\tau+1}$.   We use the following observation.

\begin{observation}\label{obs:win}
A strategy $\{d_i\mid 0\leq i\leq \tau-1\}$ is a winning strategy for \divider with $k$ agents in Rendezvous Games with Adversaries for $\tau$ steps on $G$ against Facilitator starting from $s$ and $t$ if and only if  $F$ is a set of two distinct vertices for every $P_v=(F,D)$ for  $v\in V(\mathcal{T}_G^k(\tau))$.
\end{observation}

In particular, this allows to observe the following.

\begin{observation}\label{obs:co-np}
For every fixed constant $\tau$, the problem $\tau$-\probRT is in \classCoNP.
\end{observation}

\begin{proof}
If $(G,s,t,k)$ is a no-instance of $\tau$-\probRT, then \divider has a winning strategy that allows to prevent $R$ and $J$ from meeting in at most $\tau$ steps. Then the tree $\mathcal{T}_G^k(\tau)$ can be used as a certificate. Since the tree has $n^{\Oh(\tau)}$ nodes,
we can check whether a given tree encodes a winning strategy in polynomial time, using Observations~\ref{obs:win} and \ref{obs:adj}. 
\end{proof}

\section{Dynamic separation vs. separators}\label{sec:sep}
 In this section we investigate relations between $d_G(s,t)$ and $\lambda_G(s,t)$. Given a connected graph $G$ and two vertices $s$ and $t$, it is straightforward to see that $d_G(s,t)\leq \lambda_G(s,t)$.  Indeed, if $S\subseteq V(G)\setminus\{s,t\}$ is an $(s,t)$-separator of size $k=\lambda_G(s,t)$, then \divider with $k$ agents can put then in the vertices of $S$ in the beginning of the game. Then he can use the trivial strategy that keeps the agents $D_1,\ldots,D_k$ in their positions.   However, $d_G(s,t)$ and $\lambda_G(s,t)$ can be far apart. Still, $d_G(s,t)=1$ if and only if $\lambda_G(s,t)=1$, and this is the first result of the section.
 
\begin{theorem}\label{thm:sep-one}
Let $G$ be a connected graph and let $s,t\in V(G)$. Then $d_G(s,t)=1$ if and only if $\lambda_G(s,t)=1$.
\end{theorem} 
 
\begin{proof}
As we already observed, $d_G(s,t)\leq \lambda_G(s,t)$. Hence, if $\lambda_G(s,t)=1$, then $d_G(s,t)=1$. This means that it is sufficient to show that if $d_G(s,t)=1$, then $\lambda_G(s,t)=1$. We prove this by contradiction. Assume that $\lambda_G(s,t)\geq 2$. We show that Facilitator  has a winning strategy when starting from $s$ and $t$ on $G$ against  \divider with one agent. 

Let $C$ be a cycle in $G$. For every two distinct vertices $u$ and $v$ of $C$, $C$ has two internally vertex disjoint $(u,v)$-paths $P_1$ and $P_2$ in $C$. We say that $C$ has a \emph{$(u,v)$-shortcut} if there is a $(u,v)$-path $P$ in $G-(V(C)\setminus\{u,v\})$ that is shorter than $P_1$ and $P_2$. That is,  $\ell(P)<\ell(P_1)$ and $\ell(P)<\ell(P_2)$. We say that $C$ has a \emph{shortcut} if there are 
distinct $u,v\in V(C)$ that have a $(u,v)$-shortcut. 

We   claim the following.

\begin{claim}\label{cl:step-one}
If $R$ and $J$   occupy vertices of a cycle $C$ of $G$ that has a shortcut, then Facilitator has a strategy such that in at most $\ell(C)$ steps
$R$ and $J$ are moved into vertices of a cycle $C'$ with $\ell(C')<\ell(C)$.
\end{claim}

\begin{proof}[Proof of Claim~\ref{cl:step-one}]
Suppose that $R$ and $J$  occupy  vertices $x$ and $y$ of $C$, respectively. 
Assume that a path $P$ is a $(u,v)$-shortcut for some distinct $u,v\in V(G)$. Denote by $P_1$ and $P_2$, respectively, the internally vertex disjoint $(u,v)$-paths in $C$. Let $C_1$ be the cycle of $G$ composed by $P_1$ and $P$, and let $C_2$ be the cycle composed by $P_2$ and $P$. Because $P$ is a shortcut for $C$, we have that $\ell(C_1)<\ell(C)$ and $\ell(C_2)<\ell(C)$.  If $x,y\in V(P_1)$, then $x,y\in V(C_1)$ and the claim holds trivially, since $R$ and $J$ are already 
on cycle $C_1$ with $\ell(C_1)<\ell(C)$. Symmetrically, if $x,y\in V(P_2)$, then the claim holds. Assume that this is not the case. Then $x$ and $y$ are internal vertices of $P_1$ and $P_2$ belonging to distinct paths. We assume without loss of generality that $x$ is an internal vertex of $P_1$ and $y$ is an internal vertex of $P_2$.   

Facilitator uses the following strategy. In each step, $R$ is moved along $P_1$ toward $u$, unless the next vertex is occupied by $D_1$. In the last case, $R$ stays in the current position. Similarly,  $J$  moves toward $v$ in $P_2$ whenever this is possible and stays in the current position if the way is blocked. Notice that, since the unique agent $D_1$  of \divider occupies a unique vertex in each step, at least one of the agents $R$ or $J$  moves to an adjacent vertex. Therefore, either $R$ reaches $u$ or $J$ reaches $v$ in at most $\ell(C)$ steps. If $R$ is in $u$, then $R$ and $J$ are in the vertices of $C_2$ and $\ell(C_2)<\ell(C)$. Symmetrically, if $J$ reaches $v$, then $R$ and $J$ reach  $C_1$ with $\ell(C_1)<\ell(C)$. 
\end{proof}

Next, we show that Facilitator can win if $R$ and $J$ are in a cycle without shortcuts and $D_1$ is in the same cycle.

\begin{claim}\label{cl:step-two}
If $R$ and $J$ occupy  vertices of a cycle $C$ of $G$ without a shortcut, and the unique agent $D_1$ of \divider is in a vertex of $C$ as well,
 then Facilitator has a winning strategy with at most $\ell(C)/2$ steps.
\end{claim}

\begin{proof}[Proof of Claim~\ref{cl:step-two}]
Suppose that $R$ and $J$  occupy vertices $x$ and $y$ of $C$, respectively, and that  $D_1$ occupies  $z\in V(C)$.  Denote by $P$ the unique $(x,y)$-path in $C-z$. Facilitator uses the following strategy. In every step, $R$ and $J$  move towards each other along $P$ except if they appear to occupy  adjacent vertices. In the last case, $R$ stays and $J$ moves to the vertex occupied by $R$. We show that this strategy is a feasible winning strategy. 

The proof is by induction on the length of $P$. The claim is trivial when $\ell(P)\leq 2$. Assume that $\ell(P)\geq 3$ and the claim holds for all positions $x'$, $y'$ and $z'$ of $R$, $J$ and $D_1$, respectively, if the length of the $(x',y')$-path in $C-z'$ is at most $\ell(P)-1$. 

In the first step, $R$ and $J$   move  to the neighbors  $x'$ and $y'$ of $x$ and $y$, respectively, in $P$. If $D_1$  moves to a vertex $z'\in V(C)$, then we apply the inductive assumption and, since the length of the $(x',y')$-subpath $P'$ is  $\ell(P)-2$ and $z'\notin V(P')$, obtain that the strategy of Facilitator is winning. Assume that by the first move \divider removes $D_1$ from  $C$. If $D_1$ does not return to a vertex of $C$ in $\ell(P)/2$ steps, Facilitator wins. Hence for some $h\leq \ell(P)/2$, at the $h$-th move,  $D_1$ steps back on  a vertex $z'\in V(C)$. 

\begin{figure}[ht]
\centering
\scalebox{0.7}{
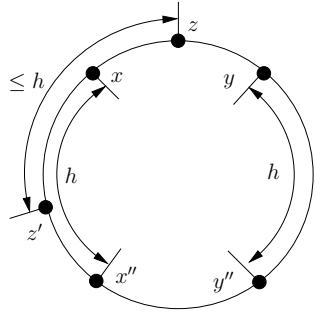}
\caption{The position placement after $h$ steps (up to symmetry).}
\label{fig:strategy}
\end{figure}

By the assumption, cycle $C$ has no shortcuts. In particular, there is no $(z,z')$-shortcut. This implies, that the length of one of the two $(z,z')$-paths in $C$ is at most $h$.  Observe that in $h$ steps, $R$ and $J$  reach vertices $x''$ and $y''$ that are at the distance $h$ in $P$ from $x$ and $y$, respectively.   Therefore (see Figure~\ref{fig:strategy}), the $(x'',y'')$-subpath $P''$ of $P$ does not contain $z'$. Since $\ell(P'')<\ell(P)$, we can apply the inductive assumption. 
This  proves that the Facilitator's strategy  is a feasible winning strategy and the claim holds. 
  
Notice that the total number of steps is $\lceil \ell(P)/2\rceil\leq \ell(C)/2$. This completes the proof.  
\end{proof}

Now we are ready to complete the proof of the theorem. 
If $s=t$ or $s$ and $t$ are adjacent, then Facilitator has a straightforward winning strategy. Assume that $s$ and $t$ are distinct and nonadjacent.
Since $\lambda_G(s,t)\geq 2$, by  Menger's theorem (see, e.g., \cite{Diestel12}), there are two internally vertex disjoint $(s,t)$-paths $P_1$ and $P_2$. The union of   these two paths forms cycle $C$. 
  If the agent $D_1$ of \divider   occupies a vertex of $C'$, then Facilitator has a winning strategy by Claim~\ref{cl:step-two}.  If $D_1$ is outside $C'$, then Facilitator moves $R$ and $J$ along $C'$ towards each other. Then either $R$ and $J$ meet or $D_1$ steps on  $C'$  at some moment. In this  case, Facilitator switches to the strategy from Claim~\ref{cl:step-two} that guarantees him to win.
\end{proof} 
 
We observed that $d_G(s,t)\leq \lambda_G(s,t)$ and, by Theorem~\ref{thm:sep-one}, $d_G(s,t)=1$ if and only if $\lambda_G(s,t)=1$. However, if $d_G(s,t)\geq 2$, then the difference betweem $\lambda_G(s,t)$ and $d_G(s,t)$ may be arbitrary. To see this, consider the following example.

\medskip
Let $p\geq 2$. 
\begin{itemize}
\item Construct a set $U=\{u_1,\ldots,u_p\}$ of pairwise adjacent vertices.  
\item Add a vertex $s$ and join $s$ with each vertex $u_i\in U$ by a path $sx_iu_i$.
\item Add  a vertex $t$ and join $t$ with each vertex $u_i\in U$ by a path $ty_iu_i$.
\end{itemize}
Denote the obtained graph by $G$ (see the left part of Figure~\ref{fig:examples}).
Observe that $\lambda_G(s,t)=p$. We show that $d_G(s,t)=2$ by demonstrating a winning strategy for \divider with two agents $D_1$ and $D_2$. Initially, $D_1$ and $D_2$ are placed in arbitrary vertices of the clique $U$. Then $D_1$ ``shadows'' $R$ and $D_2$ ``shadows'' $J$ in $U$ in the following sense. If $R$  moves to $x_i$ for some $i\in\{1,\ldots,p\}$, \divider responds by moving $D_1$ to $u_i$. Symmetrically,  if $R$ moves to $y_j$ for some $j\in\{1,\ldots,p\}$, then $D_2$ is moved to $u_j$. It is easy to verify that if \divider follows this strategy, then neither $R$ no $J$ can enter $U$. Therefore, \divider wins. Since $p$ can be arbitrary, we have 
that $\lambda_G(s,t)-d_G(s,t)=p-2$ can be arbitrary large.  
 
\begin{figure}[ht]
\centering
\scalebox{0.75}{
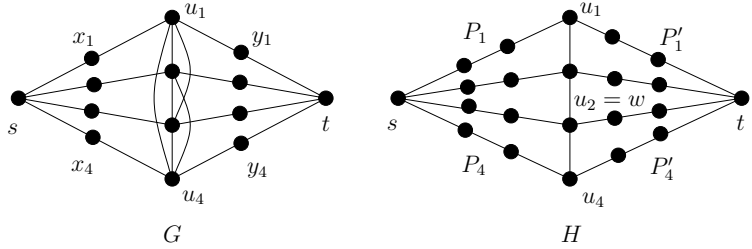}
\caption{The construction of $G$ and $H$ for $p=4$.}
\label{fig:examples}
\end{figure} 

The family of graphs $G$ for $p\geq 2$ in the above example is a family of dense graph, because $G$ contains a clique with $p$ vertices. However, exploiting the same idea as for $G$, we can show that there are sparse graphs with the same property. For this, we considered the following more complicated example.

\medskip
Let $p\geq 2$. 
\begin{itemize}
\item Construct a path $P=u_1\cdots u_p$ on $p$ vertices.  
\item Add a vertex $s$ and join $s$ with each vertex $u_i\in V(P)$ by an $(s,u_i)$-path $P_i$ of length $h=\lfloor p/2\rfloor+1$.
\item Add  a vertex $t$ and join $t$ with each vertex $u_i\in V(P)$ by an $(t,u_i)$-path $P_j'$ of length $h=\lfloor p/2\rfloor+1$.
\end{itemize}
Denote the obtained graph by $H$ (see the right part of Figure~\ref{fig:examples}). Clearly,  $\lambda_H(s,t)=p$. 
We claim that $d_H(s,t)=p$. The idea behind the winning strategy for \divider with two agents $D_1$ and $D_2$ is similar to the one from the first example:
 $D_1$ ``shadows'' $R$ and $D_2$ ``shadows'' $J$ on $P$. Let $w=u_{\lfloor p/2\rfloor}$. Initially, $D_1$ and $D_2$ are placed in $w$. Then $D_1$ is moved as follows. If $R$ moves to/stays in  $s$, then $D_1$ moves to/stays in $w$. If $R$ is moved into an internal vertex $x$ of $P_i$ for some $i\in\{1,\ldots,p\}$, then \divider responds my moving $D_1$ toward $u_i$ or keeping $D_1$ in the current position maintaining the following condition: $D_1$ is in a vertex $u_j$ at minimum distance from $w$  such that the distance between $x$ and $u_i$ in $P_i$ is more than the distance between $u_j$ and $u_i$ in $P$. The construction of the strategy for $D_2$ is symmetric.  It is easy to see that the described strategy for \divider is feasible, and the strategy allows neither $R$ no $J$ to enter a vertex of $P$. Therefore, $d_H(s,t)=2$. 
 
 Notice that the graph $H$ for each $p\geq 2$ is planar and it can be seen that the treewidth of $H$ is at most 3 (we refer to \cite{CyganFKLMPPS15,Diestel12} for the formal treewidth definition), that is,  graphs $H$ are, indeed, sparse. 
 
Our examples indicate that $\lambda_G(s,t)$ may differ from $d_G(s,t)$ if $G$ has sufficiently long induced paths and cycles. We observe that 
 $\lambda_G(s,t)=d_G(s,t)$ if $G$ belongs to graph classes that have no graphs of this type.

A graph $G$ is $\emph{$P_5$-free}$ if $G$ has no induced subgraph isomorphic to the path with 5 vertices. 

\begin{proposition}\label{ref:p-five}
If $G$ is a connected $P_5$-free graph, then for every $s,t\in V(G)$, $d_G(s,t)=\lambda_G(s,t)$.
\end{proposition}

\begin{proof}
Let $G$ be a $P_5$-free graph and let $s,t\in V(G)$. The statement is trivial if $s=t$ or $s$ and $t$ are adjacent. Assume that $s$ and $t$ are distinct nonadjacent vertices. 
Since $d_G(s,t)\leq \lambda_G(s,t)$, it is sufficient to show the opposite inequality. We prove that if \divider with $k$ agents has a winning strategy on $G$ against Facilitator starting from $s$ and $t$, then $\lambda_G(s,t)\geq k$. Let $P$ be an induced $(s,t)$-path. Since $G$ is $P_5$-free, $\ell(P)\leq 3$. Suppose  that $\ell(P)=2$, that is $P=sxt$ for some $x\in V(G)$. Then \divider has to place an agent in $x$ in the beginning of the game. Otherwise, $R$ and $J$ move to $x$ in the first step and Facilitator wins. 
Suppose that $\ell(P)=3$, that is, $P=sxyt$ for some $x,y\in V(G)$. Observe that \divider has to place an agent in $x$ or $y$ in the beginning of the game.
Otherwise, $R$ moved to $x$, $J$ moved to $y$, and Facilitator wins in the next step, as $x$ and $y$ are adjacent. Let $S$ be the set of vertices containing the agents of \divider in the beginning of the game. Clearly, $k\geq |S|$. We have that every induced $(s,t)$-path $P$ contains a vertex of $S$. This means that $S$ is an  $(s,t)$-separator. Therefore, $k\geq |S|\geq\lambda_G(s,t)$ and the claim follows.
\end{proof}

A graph $G$ is \emph{chordal} if $G$ does not contain induced cycles on at least 4 vertices, that is, if $C$ is a cycle in $G$ of length at least 4, then there is a \emph{chord}, i.e., an edge of $G$ with end-vertices  in two nonconsecutive vertices of $C$. We need some properties of chordal graphs (we refer to~\cite{BrandstadtLS99,Golumbic04} for the detailed introduction). 

It follows from the results of Gavril~\cite{Gavril74} that a graph $G$ is chordal if and only of it has a tree decomposition with every bag being a clique. Formally,
$G$ is a chordal graph if and only if there is a pair $\mathcal{T}=(T,\{X_i\}_{i\in V(T)})$, where $T$ is a tree whose every node $i$ is assigned a vertex subset $X_i\subseteq V(G)$, called a \emph{bag}, such that $X_i$ is a clique and the following holds:
\begin{itemize}
\item[(i)] $\bigcup_{i\in V(T)} X_i =V(G)$,
\item[(ii)] for every $uv\in E(G)$, there exists a node $i$ of $T$ such that $u,v\in X_i$, and
\item[(iii)] for every  $u\in V(G)$, the set $T_u = \{i\in V(T) | u\in X_i\}$, i.e., the set of nodes whose corresponding  bags contain $u$, induces
a connected subtree of $T$.
\end{itemize}
We use the following  well-known property of tree decomposition (see, e.g.,  \cite{CyganFKLMPPS15,Diestel12}). Assume that $G$ is connected. Let  $x,y\in V(T)$ be adjacent nodes of $T$ with $S=X_x\cap X_y$, 
and let $T_1$ and $T_2$ be the connected components of $T-xy$. Then for every $u\in \big(\bigcup_{i\in V(T_1)}X_i\big)\setminus S$ and every $v\in\big( \bigcup_{i\in V(T_2)}X_i\big)\setminus S$, $S$ is a $(u,v)$-separator.

\begin{proposition}\label{ref:chord}
If $G$ is a connected chordal graph, then for every $s,t\in V(G)$, $d_G(s,t)=\lambda_G(s,t)$.
\end{proposition}

\begin{proof}
Let $G$ be a chordal graph and let $s,t\in V(G)$. 
As before, let us notice
that the proposition is trivial if $s=t$ or $s$ and $t$ are adjacent, and we assume that $s$ and $t$ are distinct nonadjacent vertices. 
Recall also that  it is sufficient to show that $\lambda_G(k)\leq d_G(s,t)$. We prove that Facilitator has a winning strategy against \divider with $k$ agents if $k< \lambda_G(s,t)$. For this, we show that $R$ can reach $t$ occupied by $J$.

Since $G$ is a chordal graph, there is a tree decomposition $\mathcal{T}=(T,\{X_i\}_{t\in V(T)})$ of $G$ such that every bag $X_i$ is a clique. 
Let $i,j\in V(T)$ be nodes of $T$ such that $s\in X_i$, $t\in X_j$ and the $(i,j)$-path $P$ in $T$ has minimum length. Since $s$ and $t$ are nonadjacent, $i\neq j$.   
Let $P=i_1\cdots i_r$, where $i=i_1$ and $j=s_r$, and let $S_h=X_{i_{h-1}}\cap X_{i_h}$ for $h\in\{2,\ldots,r\}$. 
By the choice of $P$, $s\in X_i\setminus X_{i_h}$ for $h\in\{2,\ldots,r\}$ and
$t\notin X_j\setminus X_{i_h}$ for $h\in\{1,\ldots,r-1\}$. By the properties of tree decompositions, we obtain that $S_2,\ldots,S_r$ are $(s,t)$-separators. Since $\lambda_G(s,t)>k$, we have that $|S_h|>k$ for every $h\in\{2,\ldots,r\}$.

We describe the strategy of Facilitator, where $R$ is moved from $s$ to $t$ via vertices of $S_2,\ldots,S_r$. Since $|S_2|>k$, there is a vertex $v\in S_2$ that is not occupied by the agents of \divider. By the first move, Facilitator moves $R$ from $s$ to $v$. Assume now that $R$ is in a vertex $v\in S_h$ for some $h\in\{2,\ldots,r\}$. If $h=r$, then $R$ is moved to $t$. Otherwise, if $h<r$, then since $|S_{h+1}|>k$, there is $v'\in S_{h+1}$ that is not occupied by the agents of \divider. Then Facilitator either keeps $R$
 in $v$ if $v'=v$ or moves $R$ from $v$ to $v'$ otherwise.  Note that $v$ and $v'$ that are adjacent in the last case, because $S_{h},S_{h+1}\subseteq X_{i_{h+1}}$. Then we proceed from $v'$. It follows that $R$ reaches $t$ in $r$ steps. This completes the proof.
\end{proof}

The {\em chordality} of a graph is biggest smallest size of a induced cycle in it.
Clearly, chordal graphs are the graphs of chordality three. It is natural to ask whether 
for graphs of bigger chordality   the difference betweem $\lambda_G(s,t)$ and $d_G(s,t)$ may be arbitrary. In the example we gave after the Proof of Claim~\ref{cl:step-two}, we have seen that, for the graph $G$ (depicted in the left part of Figure~\ref{fig:examples}), it holds  
that $\lambda_G(s,t)-d_G(s,t)=p-2$ and is easy to see that any such $G$ has chordality five. Notice that this graph $G$ can be further enhanced so to obtain chordality four: just add a clique between the vertices in $\{x_{1},\ldots,x_{p}\}$
and a   a clique between the vertices in $\{y_{1},\ldots,y_{p}\}$.
This indicates a sharp transition of $d_{G}$ away from $\lambda_{G}$
when graphs are not chordal any more.
 
 Since $\lambda_G(s,t)$ can be computed in polynomial time by the standard maximum flow algorithms (see, e.g., the recent textbook~\cite{Williamson19}), we obtain the following corollary.
 
 \begin{corollary}\label{cor:p-five-cord}
 \probR can be solved in polynomial time on the classes of $P_5$-free and chordal graphs.
 \end{corollary} 
 
 \section{Hardness of Rendezvous Game with Adversaries}\label{sec:hardness}
 In this section, we discuss algorithmic lower bounds for \probR and \probRT.
 
 We proved in Theorem~\ref{thm:backtrack} that \probR and \probRT can be solved in $n^{\Oh(k)}$. We show that it is unlikely that the dependance on $k$ can be improved. For this, we show that both problems are \classCoW{2}-hard (i.e, it is \classW{2}-had to decide whether the input is a no-instance; in fact, we show that it is \classW{2}-hard to decide whether $d_G(s,t)\leq k$) and, therefore, cannot be solved in time $f(k)\cdot n^{\Oh(1)}$ for any  computable function $f(k)$, unless $\classFPT=\classW{2}$; the result for \probRT holds even also for  $\tau$-\probRT when  $\tau\geq 2$.  Our proof also implies that neither 
 \probR nor $\tau$-\probRT, for $\tau\geq 2$, cannot be solved in time $f(k)\cdot  n^{o(k)}$ unless $\classETH$ fails.
 
Observe that \probRT can be solved in polynomial time if $\tau=1$, because of the following straightforward observation.

\begin{observation}\label{obs:one-step}
Facilitator can with in Rendezvous Game with Adversaries game in one step on $G$ starting from $s$ and $t$ against \divider with $k$ agents if and only if one of the following holds: (i) $s=t$, (ii) $s$ and $t$ are adjacent, or (iii) $|N_G(s)\cap N_G(t)|>k$. 
\end{observation}
 
However, if $\tau\geq 2$, $\tau$-\probRT   becomes hard.

\begin{theorem}\label{thm:W-hardness}   
 \probR and $\tau$-\probRT for every constant $\tau\geq 2$ are \classCoW{2}-hard when parameterized by $k$.
 Moreover, these problems cannot be solved in time $f(k)\cdot  n^{o(k)}$ unless $\classETH$ fails.

\end{theorem} 
 
\begin{proof}
We show theorem by reducing the \textsc{Set Cover} problem. Given a universe $U$, a family $\mathcal{S}$ of subsets of $U$, and a positive integer $k$, the task of \textsc{Set Cover} is to decide whether there is a subfamily $\mathcal{S}'\subseteq \mathcal{S}$ of size at most $k$ that covers $U$, that is, every element of $U$ is in at least one set of $\mathcal{S}'$. This problem is well-known to be \classW{2}-complete when parameterized by $k$ (see, e.g.,~\cite{CyganFKLMPPS15}).

\begin{figure}[ht]
\centering
\scalebox{0.7}{
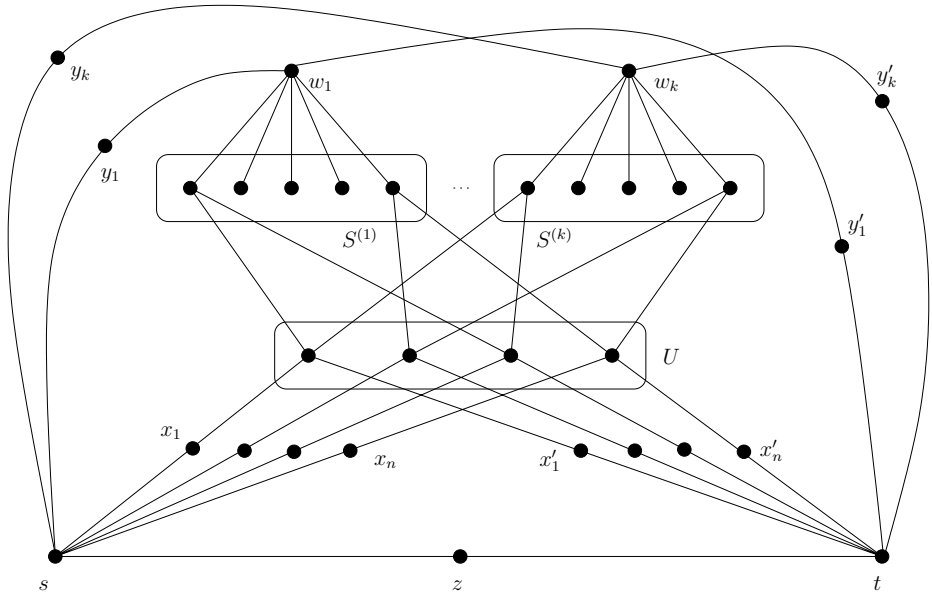}
\caption{The construction of $G$.}
\label{fig:w-hardness}
\end{figure} 

Let $(U,\mathcal{S},k)$ be an instance of \textsc{Set Cover}. Let $U=\{u_1,\ldots,u_n\}$ and $\mathcal{S}=\{S_1,\ldots,S_m\}$.   We construct the graph $G$ as follows (see Figure~\ref{fig:w-hardness}).
\begin{itemize}
\item Construct a set of $n$ vertices $U=\{u_1,\ldots,u_n\}$ corresponding to the universe.
\item For every $i\in\{1,\ldots,k\}$, construct a set of $m$ vertices $S^{(i)}=\{s_1^{(i)},\ldots,s_m^{(i)}\}$; each $S^{(i)}$ corresponds to a copy of $\mathcal{S}$.
\item For every $i\in\{1,\ldots,k\}$, $h\in\{1,\ldots,m\}$ and $h\in\{1,\ldots,n\}$, make $s_j^{(i)}$ and $u_h$ adjacent if the element of the universe $u_h$ is in 
$S_j\in\mathcal{S}$. 
\item For every $i\in\{1,\ldots,k\}$, construct a vertex $w_i$ and make it adjacent to $s_1^{(i)},\ldots,s_m^{(i)}$.
\item Construct two vertices $s$ and $t$.
\item For every $h\in \{1,\ldots,n\}$, join $s$ and $u_h$ by a path $sx_hu_h$ and joint $u_h$ and $t$ by a path $u_hx_h't$.
\item For every $i\in\{1,\ldots,k\}$, join $s$ and $w_i$ by a path $sy_iw_i$ and join $w_i$ and $t$ by a path $w_iy_i't$.
\item Construct a vertex $z$ and make it adjacent to $s$ and $t$.
\end{itemize}
 
We show that if $(U,\mathcal{S},k)$ is a yes-instance of \textsc{Set Cover}, then \divider with $k+1$ agents can win  in Rendezvous Game with Adversaries.
Let $\mathcal{S}'=\{S_{i_1},\ldots,S_{i_k}\}$ be a set cover; we assume  without loss of generality that $\mathcal{S}'$ has size exactly $k$. 
We describe a winning strategy for \divider with the agents $D_1,\ldots,D_{k+1}$. Initially, \divider puts $D_j$ in the vertex $s_{i_j}^{(j)}$ for each $j\in\{1,\ldots,k\}$, and $D_{k+1}$ is placed in $z$. Then the following strategy is used. The agents $D_1,\ldots,D_{k+1}$ are keeping their position until either $R$ or $J$ are moved from $s$ or $t$, respectively. Assume by symmetry that $R$ is moved by Facilitator  from $s$ ($J$ can either move or stay in $t$).
If $R$ is moved from $s$ to $y_j$ for some $j\in\{1,\ldots,k\}$, then \divider moves $D_j$ from $s_{i_j}^{(j)}$ to $w_j$ and $D_{k+1}$ is moved from $z$ to $s$.  
Notice that $R$ is in the vertex $y_j$ of degree two and both neighbors of $y_j$ are occupied by the agents of \divider. Hence, $R$ cannot move and $J$ cannot reach $y_j$. This implies that \divider wins by keeping the agents in their current positions.  
Assume that   $R$ is moved from $s$ to $x_h$ for some $h\in\{1,\ldots,n\}$. Since $\mathcal{S}'$ is a set cover, there is $j\in\{1,\ldots,k\}$ such that the element of the universe $u_h\in S_{i_j}$. Then $D_j$ is in the vertex $s_{i_j}^{(j)}$ that is adjacent to the vertex $u_h$. 
\divider responds by moving   $D_j$ from $s_{i_j}^{(j)}$ to $u_h$ and $D_{k+1}$ is moved from $z$ to $s$.
Now we have that $R$ is blocked in $x_h$ by the agents in $s$ and $u_h$. This means that \divider wins.

Next, we claim that if  $(U,\mathcal{S},k)$ is a no-instance of \textsc{Set Cover}, then Facilitator wins in at most two steps against \divider with $k+1$ agents. Assume that \divider completed the initial placement of the agents. If $z$ is not occupied, then Facilitator moves $R$ and $J$ to $z$ and wins in one step. Assume that $z$ is occupied by $D_{k+1}$. If there is $i\in\{1,\ldots,k\}$ such that there is no agent of \divider in a vertex of $N_G[w_i]$, then Facilitator moves $R$ to $y_i$ and $J$ to $y_i'$ by the first move. Since \divider has no agents in $N_G[w_i]$, for any his move, $w_i$ remains unoccupied by his agents. Therefore, Facilitator can move $R$ and $J$ to $w_i$ and win in two steps. Suppose from now that for every $i\in\{1,\ldots,k\}$, $D_i$ is in $N_G[w_i]$. Because \divider has $k+1$ agents, this means that for every $h\in\{1,\ldots,h\}$,
$x_h$, $u_h$ and $x_h'$ are not occupied by the agents of \divider, and for every $i\in\{1,\ldots,k\}$, at most one agent is in $S^{(i)}$.  Let $X$ be the set of vertices of $\bigcup_{i=1}^kS^{(i)}$ occupied by the agents of \divider in the beginning of the game. Since $|X|\leq k$ and  $(U,\mathcal{S},k)$ is a no-instance of \textsc{Set Cover}, there is $h\in \{1,'\ldots,n\}$ such that the vertices $N_H[u_h]$ are not occupied by the agents if \divider. Therefore, Facilitator can move $R$ from $s$ to $x_h$ and then to $u_h$ and, symmetrically, move $J$ from $t$ to $x_h'$ and then to $u_h$. We obtain that $R$ and $J$ meet in $u_h$ in two steps, that is, Facilitator wins in two steps. 

These arguments imply that $(U,\mathcal{S},k)$ is a yes-instance of \textsc{Set Cover} if and only if $(G,s,t,k+1)$ is a no-instance of \probR. This mens that \probR is \classCoW{2}-hard. For $\tau$-\probRT for $\tau\geq 2$, notice that if $(U,\mathcal{S},k)$ is a no-instance of \textsc{Set Cover}, then Facilitator can win in at most two steps against \divider with $k+1$ agents, and if $(U,\mathcal{S},k)$ is a yes-instance, then \divider has a strategy that prevents Facilitator from winning in any number of steps. It follows that $\tau$-\probRT is \classCoW{2}-hard for any fixed $\tau\geq 2$.

For the second part of the claim of Theorem~\ref{thm:W-hardness}, we use the results of Chen et al.~\cite{ChenHKX06}, see also  \cite[Corollary~14.23]{CyganFKLMPPS15}.   In particular, they proved that 
 \textsc{Set Cover} cannot be solved in time $f(k)\cdot (n+m)^{o(k)}$ unless $\classETH$ fails. To show \classCoW{2}-hardness of \probR and $\tau$-\probRT, for $\tau\geq 2$, we constructed a polynomial reduction and the obtained parameter for \probR and $
\tau$-\probRT is $k+1$, i.e., is linear in the input parameter for \textsc{Set Cover}. Thus, any algorithm for \probR or $\tau$-\probRT, for $\tau\geq 2$, with running time $f(k)\cdot n^{o(k)}$ would imply the algorithm for \textsc{Set Cover} with running time $f(k)\cdot  (n+m)^{o(k)}$.
\end{proof} 
  We proved Theorem~\ref{thm:W-hardness} by giving a polynomial reduction from \textsc{Set Cover}. Since \textsc{Set Cover} is \classNP-complete (see~\cite{GareyJ79}), we obtain the following corollary using Observation~\ref{obs:co-np}.
 
\begin{corollary}\label{cor:np-hardness}
 $\tau$-\probRT is \classCoNP-complete for every fixed constant $\tau\geq 2$.
\end{corollary} 
 
 Using the reduction from the proof of Theorem~\ref{thm:W-hardness}, we can conclude that \probR and $\tau$-\probRT, for $\tau\geq 2$, are \classCoNP-hard. However, the general problems are harder.
 
 \begin{theorem}\label{thm:pspace}
\probR and \probRT are \classPSPACE-hard. 
 \end{theorem} 
 
 \begin{proof}
 We show that \probRT is \classPSPACE-hard and then explain how to modify our reduction for \probR. We prove that it is \classPSPACE-hard to decide whether \divider can win in at most $\tau$ steps in Rendezvous Game with Adversaries.
  
The reduction is from the \textsc{Quantified Boolean Formula in Conjunctive Normal Form} (QBF) problem with alternating quantifiers that is well-known to be \classPSPACE-complete (see, e.g.,~\cite{GareyJ79}). The task of QBF is, given $2n$ Boolean variables $x_1,\ldots,x_{2n}$ and $m$ clauses $C_1,\ldots,C_m$, where every $C_i$ is a disjunction of literals over the variables, to decide whether the formula
 $$\varphi= \forall x_1\exists x_2\ldots \forall x_{2n-1}\exists x_n[C_1\wedge\ldots\wedge C_m]$$
 evaluates \true.   
 
 \begin{figure}[ht]
\centering
\scalebox{0.7}{
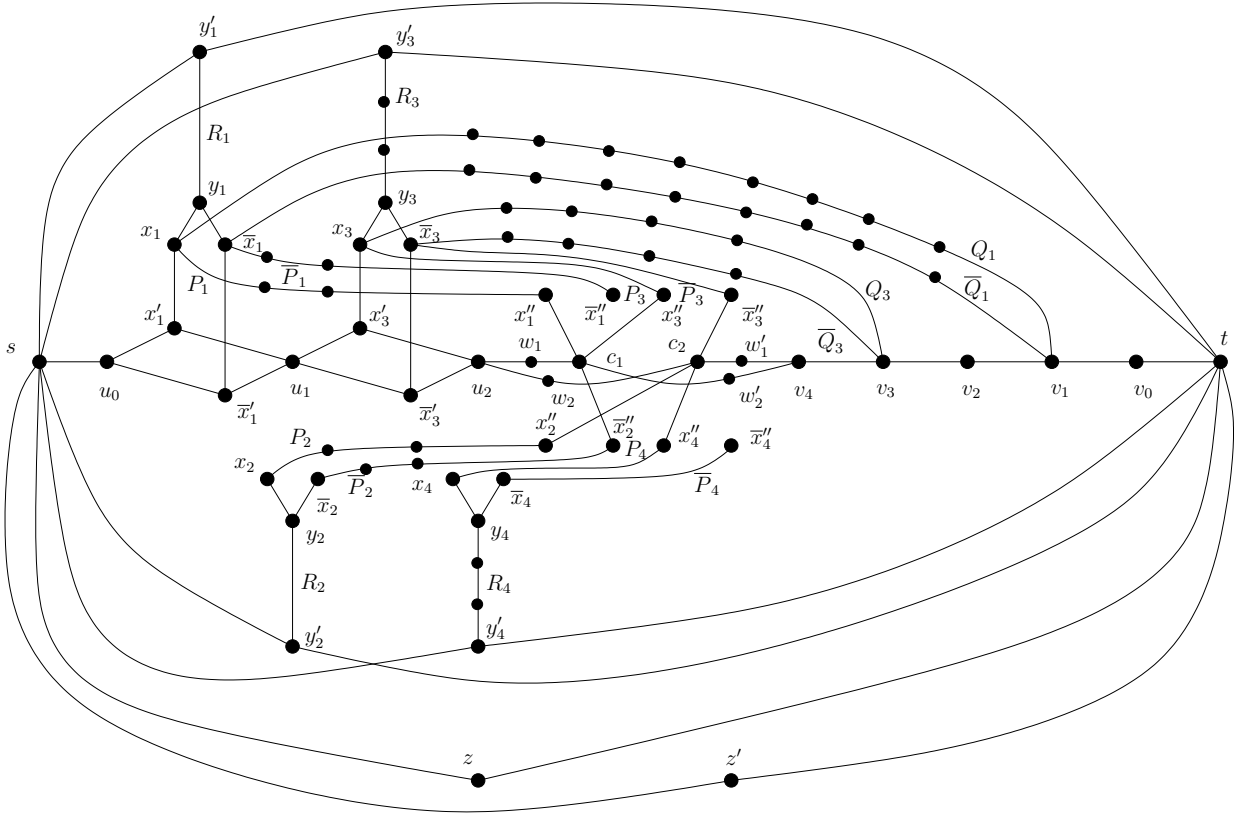}
\caption{The construction of $G$ for $\varphi=\forall x_1\exists x_2\forall x_3\exists x_4[(x_1\vee\overline{x}_2\vee x_3)\wedge(x_2\vee\overline{x}_3\vee x_4)]$.}
\label{fig:ps-hardness}
\end{figure}

 Given a formula $\varphi= \forall x_1\exists x_2\ldots \forall x_{2n-1}\exists x_n[C_1\wedge\ldots\wedge C_m]$, we construct the graph $G$ as follows (see Figure~\ref{fig:ps-hardness}). 
\begin{itemize}
\item Construct $m$ vertices $c_1,\ldots,c_m$ corresponding to the clauses of $\varphi$.
\item Construct vertices $s$, $u_0,\ldots,u_n$, and $x_{2i-1},\overline{x}_{2i-1}$ for $i\in\{1,\ldots,n\}$, and 
for each $i\in\{1,\ldots,n\}$, make $x_{2i-1},\overline{x}_{2i-1}$ adjacent to $u_{i-1}$ and $u_i$.
\item Make $s$ and $u_0$ adjacent, and join $u_n$ with $c_1,\ldots,c_m$ by paths $u_nw_jc_j$ for $j\in\{1,\ldots,m\}$.
\item Construct $2n+1$ vertices $v_0,\ldots,v_{2n}$ and construct the path $tv_0\cdots v_{2n}$. Then join $v_{2n}$ with $c_1,\ldots,c_m$ by paths $v_{2n}w_j'c_j$ for $j\in\{1,\ldots,m\}$.
\item Construct two vertices $z$ and $z'$, and make them adjacent to $s$ and $t$. 
\item For every $i\in\{1,\ldots,2n\}$, construct vertices $x_i,\overline{x}_i,x_i'',\overline{x}_i''$ and $y_i,y_i'$, and then make $y_i$ adjacent to $x_i,\overline{x}_i$ and make $y_i'$ adjacent to $s$ and $t$. 
\item For every $i\in\{1,\ldots,n\}$, 
\begin{itemize}
\item construct $(x_{2i-1},x_{2i-1}'')$ and  $(\overline{x}_{2i-1},\overline{x}_{2i-1}'')$-paths $P_{2i-1}$ and $\overline{P}_{2i-1}$, respectively, of length $2(n-i)+1$,
\item construct $(x_{2i-1},v_{2i-1})$ and  $(\overline{x}_{2i-1},v_{2i-1})$-paths $Q_{2i-1}$ and $\overline{Q}_{2i-1}$, respectively, of length $4(n-i)+5$,
\item construct a $(y_{2i-1},y_{2i-1}')$-path $R_{2i-1}$ of length $2i-1$.
\end{itemize}
\item For every $i\in\{1,\ldots,n\}$, 
\begin{itemize}
\item construct $(x_{2i},x_{2i}'')$ and  $(\overline{x}_{2i},\overline{x}_{2i}'')$-paths $P_{2i}$ and $\overline{P}_{2i}$, respectively, of length $2(n-i)+1$,
\item construct a $(y_{2i},y_{2i}')$-path $R_{2i}$ of length $2i-1$.
\end{itemize}
\end{itemize} 
 This completes the construction of $G$. We define $\tau=2n+3$ and $k=2n+2$.
 
 We claim that $\varphi$ evaluates \true if and only if \divider with $k$ agents has a winning strategy that prevents $R$ and $J$ from meeting in at most $\tau$ steps.
 
 Assume that $\varphi=\true$. We describe a winning strategy for \divider.
 
 Assume that after the $i$-th move of Facilitator $R$ and $J$ are occupying some vertices $a$ and $b$ and the agents of \divider are in the vertices of a set $X$. If $\dist_{G-X}(a,b)>2\tau-2i$, then \divider wins by keeping the agents in  their current positions, because $R$ and $J$ are unable to meet in the remaining $\tau-i$ moves. In this case, we say that \divider has a \emph{trivial winning strategy}.
 
\divider has $k=2n+2$ agents. We place $D_i$ in $y_i'$ for $i\in\{1,\ldots,2n\}$. The remaining two agents $D_{2n+1}$ and $D_{2n+2}$ are placed in $z$ and $z'$, respectively. For $i\geq 1$, we use $X_i$ to denote the set of vertices occupied by the agents of \divider after the $i$-th step of the game; $X_0=\{y_1',\ldots,y_{2n}'\}\cup\{z,z'\}$.  
 
 Observe that $\dist_{G-X_0}(s,t)=2\tau$. This means that if either $R$ or $J$ is  not moved, then \divider wins by the trivial winning strategy.  We assume that this is not the case and $R$ is moved to $u_0$ and $J$ is moved to $v_0$. \divider responds by moving $D_{2n+1}$ and $D_{2n}$ from $z$ and $z'$ to $s$ and $t$, respectively; these agents remain in $s$ and $t$ until the end of the game and the only role of them is to prevent $R$ and $J$ from entering these vertices. The agents $D_1,\ldots,D_{2n}$ are moved to the neighbors of $y_1',\ldots,y_{2n}'$ in the paths $R_1,\ldots,R_{2n}$, respectively. 
 
The general idea of the reduction is that by the further $2n$ steps, the players define the values of the Boolean variables $x_1,\ldots,x_{2n}$, and the values of the variables $x_1,x_3,\ldots,x_{2n-1}$ are chosen by Facilitator, and \divider chooses the values of $x_2,x_4,\ldots,x_{2n}$. 
\divider chooses the value of his variables to achieve $\psi=C_1\wedge\ldots\wedge C_m=\true$.
To describe this process, we show 
inductively for $i=0,\ldots,n$ that after the $2i+1$-th step of the game, either \divider wins by the trivial strategy or the following configuration is maintained. 

\begin{itemize}
\item The values of the variables $x_j$ for $j\leq 2i$ are chosen and the values of the variables $x_j$ for $j>2i$ are unassigned. Moreover, the values of $x_1,\ldots,x_{2i}$ are chosen in such a way that $\varphi$ evaluates $\true$ if the values of   $x_1,\ldots,x_{2i}$ are constrained by the choice.
\item $R$ is in $u_{i}$ and $J$ is in $v_{2i}$.
\item For $j\in\{i+1,\ldots,n\}$, $D_{2j-1}$ and $D_{2j}$ are on the paths $R_{2j-1}$ and $R_{2j}$, respectively, at distance $\ell(R_{2j-1})-2i-1=\ell(R_{2j})-2i-1=2(j-i)-2$ from $y_{2j-1}$ and $y_{2j}$, respectively; in particular,  $D_{2i+1}$ and $D_{2i+2}$ are in $y_{2i+1}$ and $y_{2i+2}$, respectively, if $i<n$. 
\item For $j\in\{1,\ldots,i\}$, 
\begin{itemize}
\item if the variable $x_{2j-1}=\true$, then $D_{2j-1}$ is on the path $P_{2j-1}$ at distance $\ell(P_{2j-1})+\ell(R_{2j-1})-2i=2(n-i)$ from $x_{2j-1}''$,
\item if the variable $x_{2j-1}=\false$, then $D_{2j-1}$ is on the path $\overline{P}_{2j-1}$ at distance $\ell(\overline{P}_{2j-1})+\ell(R_{2j-1})-2i=2(n-i)$ from $\overline{x}_{2j-1}''$, 
\item if the variable $x_{2j}=\true$, then $D_{2j}$ is on the path $P_{2j}$ at distance $\ell(P_{2j})+\ell(R_{2j})-2i=2(n-i)$ from $x_{2j}''$,
\item if the variable $x_{2j}=\false$, then $D_{2j}$ is on the path $\overline{P}_{2j}$ at distance $\ell(\overline{P}_{2j})+\ell(R_{2j})-2i=2(n-i)$ from $\overline{x}_{2j}''$. 
\end{itemize}
\end{itemize}
 
It is straightforward to verify that the claim holds for $i=0$. Assume inductively that the claim holds for $0\leq i<2n$. We show that either \divider wins by the trivial strategy applied from the steps $2i+2$ or $2i+3$, or  the configuration is maintained for $i'=i+1$. 

 Observe that $\dist_{G-X_{2i+1}}(u_i,v_{2i})=2(\tau-2i-1)$. Therefore, if either of the agents of Facilitator remain in their old position, then \divider wins by the trivial strategy. Therefore, both $R$ and $J$ has to move. Moreover, they have to move along a shortest $(u_i,v_{2i})$-path in $G-X_{2i+1}$. Hence, Facilitator moves $J$ from $v_{2i}$ to $v_{2v+1}$ and $R$ is moved either to $x_{2i+1}'$ or to $\overline{x}_{2i+1}'$. If $R$ is moved to $x_{2i+1}'$, then we assign the variable $x_{2i+1}=\true$, and $x_{2i+1}=\false$ otherwise. \divider responds by moving $D_{2i+1}$ from $y_{2i+1}$ to $x_{2i+1}$ if $R$ is in $x_{2i+1}'$, and $D_{2i+1}$ is moved to $\overline{x}_{2i+1}$ if $R$ is in $\overline{x}_{2i+1}'$. For $D_{2i+2}$, 
 \divider chooses one of the vertices $x_{2i+2}$ and $\overline{x}_{2i+2}$ and moves the agent there. By this move, \divider selects the value of the Boolean variable $x_{2i+2}$, and  $x_{2i+2}=\true$ if $D_{2i+2}$ in $x_{2i+2}$, and $x_{2i+2}=\false$ otherwise.  Note that \divider knows the values of $x_1,\ldots,x_{2i+1}$ and selects the move for $D_{2i+1}$ to ensure that the final value of $\psi=\true$.  The agents $D_h$ for $h\in\{1,\ldots,2n\}$ such that $h\neq 2i+1,~2i+2$ are moved to adjacent vertices in the corresponding paths.
For $j\in\{i+2,\ldots,n\}$, $D_{2j-1}$ and $D_{2j}$ are moved along $R_{2j-1}$ and $R_{2j}$ toward $y_{2j-1}$ and $y_{2j}$, respectively. 
  For $j\in\{1,\ldots,i\}$,   $D_{2j-1}$ and $D_{2j}$ are moved along $P_{2j-1}$ ($\overline{P}_{2j-1}$) and $P_{2j}$ ($\overline{P}_{2j}$) toward   $x_{2j-1}''$ ($\overline{x}_{2j-1}''$) and  $x_{2j}''$ ($\overline{x}_{2j}''$), respectively. 
  
 Assume without loss of generality that $R$ occupies $x_{2i+1}'$, because the case when $R$ is in $\overline{x}_{2i+1}$ is symmetric.   We have that $\dist_{G-X_{2i+2}}( x_{2i+1}',v_{2i+1})=2(\tau-2i-2)$. Hence, if $R$ of $J$ are not moved toward each other by the next step, \divider wins by the trivial strategy. Assume that this is not the case. Recall that $D_{2i+1}$ is in $x_{2+1}$.  We conclude that $R$  is moved to $u_{i+1}$ and $J$ is moved to $v_{2i+2}$.  \divider responds as follows.
For $j\in\{i+2,\ldots,n\}$, $D_{2j-1}$ and $D_{2j}$ are moved along $R_{2j-1}$ and $R_{2j}$ toward $y_{2j-1}$ and $y_{2j}$, respectively. 
  For $j\in\{1,\ldots,i+1\}$,   $D_{2j-1}$ and $D_{2j}$ are moved along $P_{2j-1}$ ($\overline{P}_{2j-1}$) and $P_{2j}$ ($\overline{P}_{2j}$) toward   $x_{2j-1}''$ ($\overline{x}_{2j-1}''$) and  $x_{2j}''$ ($\overline{x}_{2j}''$), respectively. We obtain that for $i'=i+1$, the players are in the required configuration. 
 
By the above claim, we have that after $2n+1$ steps of the game either \divider wins by the trivial strategy applied from some step or the following configuration is achieved:
\begin{itemize}
\item The values of the Boolean variables $x_1,\ldots,x_{2n}$ are chosen and $\psi=\true$ for them.
\item $R$ is in $u_n$ and $J$ is in $v_{2n}$.
\item For $i\in\{1,\ldots,2n\}$, $D_i$ is in $x_i''$ if $x_i=\true$ and $D_i$ is in $\overline{x}_i''$ otherwise.
\end{itemize} 
  
Since $\dist_{G-X_{2n+1}}(u_n,v_{2n})=4$, the only possibility for Facilitator to win in two steps is to move $R$ and $J$ toward each other along the path $u_nw_hc_hw_h'v_{2n}$ for some $h\in \{1,\ldots,m\}$. Otherwise, \divider wins by the trivial strategy. Assume that $R$ is moved to $w_h$ for some $h\in\{1,\ldots,m\}$ 
in the next step. Recall that $\psi=\true$. Therefore, the clause $C_j$ contains a literal $x_i$ or $\overline{x}_i$ for some  $i\in\{1,\ldots,2n\}$ with the value $\true$. Assume that $C_h$ contains $x_i$ as the other case is symmetric. We have that the vertex $x_i''$ is occupied by $D_i$ and $x_i''$ is adjacent to $c_h$ in $G$. \divider responds to the moving $R$ to $w_h$ by moving $D_i$ to $c_h$. This prevents $R$ and $J$ from meeting in the next step. Therefore, \divider wins.
This concludes the proof of the claim that if $\varphi$ evaluates $\true$, then  \divider with $k$ agents has a winning strategy in the game with $\tau$ steps.

 Our next aim is to show that if $\varphi$ evaluates \false, then Facilitator can win in $\tau$ steps.
 
 It is convenient to define a special strategy for Facilitator that can be applied after some step.  Assume that after the $i$-th move of Facilitator $R$ and $J$ are occupying some vertices $a$ and $b$ and $G$ has an $(a,b)$-path $L$ whose length is at most $2(\tau-i)$ and the internal vertices of $L$ have degree two in $G$. If 
 there are no agent of \divider occupying a vertex of $L$, then Facilitator wins in at most $\tau-i$ remaining steps by moving $R$ and $J$ along $L$ toward each other (except if $R$ and $J$ are in adjacent vertices; then $R$ moves to the vertex occupied by $J$). If Facilitator can win this way, we say that Facilitator has a \emph{trivial winning strategy}. 
 
As above, for $i\in\{1,\ldots,2n\}$, we use $X_i$ to denote the set of vertices occupied by the agents of \divider after $i$-ith step of the game and $X_0$ is the set of vertices
occupied in the beginning of the game.

Initially, $R$ is in $s$ and $J$ is in $t$. If there is a vertex $a\in N_G(s)\cap N_G(t)$ such that $a\notin X_0$, then Facilitator wins in one step by moving $R$ and $J$ to $a$.
Hence, we assume that $N_G(s)\cap N_G(t)\subseteq X_0$. Since $|N_G(s)\cap N_G(t)|$, $X_0=N_G(s)\cap N_G(t)=\{y_1',\ldots,y_{2n}'\}\cup\{z,z'\}$ and each vertex of $X_0$ is occupied by exactly one agent of \divider.   Let $D_i$ be in $y_i'$ for $i\in\{1,\ldots,2n\}$ and let the remaining two agents $D_{2n+1}$ and $D_{2n+2}$ be in $z$ and $z'$, respectively.

The idea behind the strategy of Facilitator is that $R$ and $J$ are moved towards each other along the paths containing the vertices $u_0,\ldots,u_n$ and $v_0,\ldots,v_{2n}$ with the aim to meet in some vertex $c_h$. The trajectory of $R$ goes trough vertices $x_i'$ and $\overline{x}_i'$ and the choice between these vertices defines the value of the variable $x_i$. On the way, Facilitator forces \divider to behave  in a certain way as, otherwise, Facilitator can win by the trivial strategy using paths $Q_i$ or $\overline{Q}_i$.

 It is convenient to sort-out the agents of \divider whose movements are irrelevant for the strategy of Facilitator. We say that an agent $D_j$ is \emph{out of game} if $D_j$ cannot block any shortest path between the vertices occupied by $R$ and $J$ in $G-X_0$. Formally,  assume that after the $i$-th step of the game $R$ is in a vertex $a$ and $J$ is in $b$, and let $d$ be the vertex occupied by $D_j$.  We say that $D_j$ is \emph{out of game} after the $i$-step of the game if for every shortest $(a,b)$-path $L$ in $G-X_0$ and every $e\in V(L)$, it holds that
 \begin{itemize}
 \item  $\dist_L(a,e)\leq\dist_G(d,e)$ if $\dist_L(a,e)\leq \dist_L(b,e)$,
 \item  $\dist_L(b,e)\leq \dist_G(d,e)$, otherwise.
  \end{itemize}
Note that if $D_j$ is out of game after the $i$-th step, then $D_j$ is out of game for all subsequent steps, because $R$ and $J$ are moving toward each other  along some shortest path between the vertices occupied by them. In particular, $D_{2n+1}$ and $D_{2n}$ are out of game from the beginning. 
 
By the first step, Facilitator moves $R$ to $u_0$ and $J$ is moved from $t$ to $v_0$. Further, $R$ and $J$ move towards each other. The moves of the players define the values of the Boolean variables $x_1,\ldots,x_{2n}$. By his moves, Facilitator consecutively chooses the values of $x_1,x_3,\ldots,x_{2n-1}$ and \divider selects the values of $x_2,x_4,\ldots,x_{2n}$. Facilitator aim to achieve $\psi=C_1\wedge\ldots\wedge C_m=\false$.
To describe the strategy,  we show 
inductively for $i=0,\ldots,n$ that after the $2i+1$-th step of the game, either Facilitator wins by the trivial strategy or the following configuration is maintained. 

\begin{itemize}
\item The values of the variables $x_j$ for $j\leq 2i$ are chosen and the values of the variables $x_j$ for $j>2i$ are unassigned. Moreover, the values of $x_1,\ldots,x_{2i}$ are chosen in such a way that $\varphi$ evaluates $\false$ if the values of   $x_1,\ldots,x_{2i}$ are constrained by the choice.
\item $R$ is in $u_{i}$ and $J$ is in $v_{2i}$.
\item For $j\in\{i+1,\ldots,n\}$, 
\begin{itemize}
\item either $D_{2j-1}$ is out of game or $D_{2j-1}$ in on the path $R_{2j-1}$  at distance $\ell(R_{2j-1})-2i-1=2(j-i)-2$ from $y_{2j-1}$, 
\item either $D_{2j}$ is out of game or  $D_{2j}$ is on the path $R_{2j}$ at distance $\ell(R_{2j})-2i-1=2(j-i)-2$ from $y_{2j}$. 
\end{itemize}
\item For $j\in\{1,\ldots,i\}$, 
\begin{itemize}
\item if the variable $x_{2j-1}=\true$, then either $D_{2j-1}$ is out of game or $D_{2j-1}$ is on the path $P_{2j-1}$ at distance $\ell(P_{2j-1})+\ell(R_{2j-1})-2i=2(n-i)$ from $x_{2j-1}''$,
\item if the variable $x_{2j-1}=\false$, then either $D_{2j-1}$ is out of game or $D_{2j-1}$ is on the path $\overline{P}_{2j-1}$ at distance $\ell(\overline{P}_{2j-1})+\ell(R_{2j-1})-2i=2(n-i)$ from $\overline{x}_{2j-1}''$, 
\item if the variable $x_{2j}=\true$, then either $D_{2j}$ is out of game or $D_{2j}$ is on the path $P_{2j}$ at distance $\ell(P_{2j})+\ell(R_{2j})-2i=2(n-i)$ from $x_{2j}''$,
\item if the variable $x_{2j}=\false$, then either $D_{2j}$ is out of game or $D_{2j}$ is on the path $\overline{P}_{2j}$ at distance $\ell(\overline{P}_{2j})+\ell(R_{2j})-2i=2(n-i)$ from $\overline{x}_{2j}''$. 
\end{itemize}
\end{itemize}

The construction of $G$ immediately implies that the claim holds for $i=0$. Assume inductively that the claim holds for $0\leq i<2n$. We show that either Facilitator wins 
by the trivial strategy applied from the steps $2i+2$ or $2i+3$, or  the configuration is maintained for $i'=i+1$. 

By the $2i+2$-th move, Facilitator moves $J$ to $v_{2i+1}$ and $R$ is moved either to $x_{2i+1}'$ or to $\overline{x}_{2i+1}'$. Note that Facilitator cannot prevent these moves, because of our assumption about  the configuration of the positions of the players and the observation that $D_{2n+1}$ and $D_{2n}$ are out of game.
If $R$ is moved to $x_{2i+1}'$, then we assign the variable $x_{2i+1}=\true$, and $x_{2i+1}=\false$ otherwise. Assume that $R$ is moved to $x_{2i+1}'$ (the other case is symmetric). If no agent of  \divider is moved to $x_{2i+1}$, then by the next moves $R$ is moved to $x_{2i+1}$ and $J$ is moved along the path $Q_{2i+1}$ toward $R$. 
Because the vertices of $Q_{2i+1}$ are not occupied by the agents of \divider, Facilitator wins by the trivial strategy. Since only $D_{2i+1}$ can move into $x_{2i+1}$, we assume that $D_{2i+1}$ is moved to this vertex. Symmetrically, we assume that if $R$ is moved to $\overline{x}_{2i+1}'$, then $D_{2i+1}$
 is moved to $\overline{x}_{2i+1}$. 

Observe that if $D_{2i+2}$ is out of game, then no agent of \divider can be moved to either $x_{2i+2}$ or $\overline{x}_{2i+2}$. Otherwise, $D_{2i+2}$ is in $y_{2i+2}$. If the agent is not moved to neither  $x_{2i+2}$ nor $\overline{x}_{2i+2}$, $D_{2i+2}$ is out of game. In all these cases, the value of the Boolean variable $x_{2i+2}$ is defined arbitrarily. Otherwise, if $D_{2i+2}$ is moved to  $x_{2i+2}$, then we set $x_{2i+2}=\true$, and if $D_{2i+2}$ is moved to  $\overline{x}_{2i+2}$, then we set $x_{2i+2}=\false$.
Consider the agents $D_h$ for $h\in\{1,\ldots,2n\}$ such that $h\neq 2i+1,~2i+2$ are not out of game. 
If such an agent $D_{2j-1}$ ( $D_{2j}$, respectively) for $j\in\{i+2,\ldots,n\}$ is not moved along $R_{2j-1}$  toward $y_{2j-1}$ (along $R_{2j}$ toward $y_{2j}$, trespectively), $D_{2j-1}$ ($D_{2j}$, respectively) is out of game.
Similarly, if such an agent    $D_{2j-1}$ ($D_{2j}$, respectively) for  $j\in\{1,\ldots,i\}$ is not moved along his current path $P_{2j-1}$ or $\overline{P}_{2j-1}$ ($P_{2j}$ or $\overline{P}_{2j}$, respectively) toward   $x_{2j-1}''$ or $\overline{x}_{2j-1}''$ ($x_{2j}''$ or $\overline{x}_{2j}''$), respectively, this agent is placed out of game. 

Now we consider the step $2i+3$. By symmetry, we assume without loss of generality that $R$ is in  $x_{2i+1}'$ (the case when $R$ is in $\overline{x}_{2i+1}$ is symmetric).   
Then Facilitator moves $R$ to $u_{i+1}$ and $J$ to $v_{2i+2}$. For each agent $D_h$ that is not out of game, observe that $D_h$ is placed on some path: 
for $j\in\{i+2,\ldots,n\}$, $D_{2j-1}$ and $D_{2j}$ are in  $R_{2j-1}$ and $R_{2j}$, respectively, and for $j\in\{1,\ldots,i+1\}$,   $D_{2j-1}$ and $D_{2j}$ are in  $P_{2j-1}$ ($\overline{P}_{2j-1}$) and $P_{2j}$ ($\overline{P}_{2j}$), respectively. If they do not move along these paths toward $y_{2(i+1)},\ldots,y_{2n}$ and the vertices $x_{h}''$ or $\overline{x}_h''$ for $h\leq 2i$, then they are out of game. 
We obtain that for $i'=i+1$, the players are in the required configuration. 

By the above claim, we have that after $2n+1$ steps of the game either Facilitator wins by the trivial strategy applied from some step or the following configuration is achieved:
\begin{itemize}
\item The values of the Boolean variables $x_1,\ldots,x_{2n}$ are chosen and $\psi=\false$ for them.
\item $R$ is in $u_n$ and $J$ is in $v_{2n}$.
\item For each $j\in\{1,\ldots,m\}$, the vertices $w_j$, $c_j$ and $w_j'$ are not occupied by the agents of \divider.
\item For $i\in\{1,\ldots,2n\}$, if $x_i''$ is occupied by an agent of \divider, then the vertex is occupied by $D_i$ and the value of the variable $x_{i}=\true$, and 
if $\overline{x}_i''$ is occupied by an agent of \divider, then the vertex is occupied by $D_i$ and the value of the variable $x_{i}=false$.
\end{itemize} 

Since $\psi=\false$, there is $j\in\{1,\ldots,m\}$ such that $C_j=\false$. Then for $c_j$, we have that the vertices of $N_G[c_j]$ are not occupied by the agents of \divider.
This mean that Facilitator wins by the next two moves: $R$ is moved to $w_j$ and then to $c_j$, and $R$ is moved to $w_j'$ and then to $c_j$. It follows that Facilitator wins on $G$ in at most $\tau$ steps.

This concludes the proof of \classPSPACE-hardness for \probRT. 

\medskip
The proof for \probR is similar but more complicated. Observe that in the winning strategy for \divider for the case $\varphi=\true$, it is crucial that $R$ and $J$ are forced to move toward each other along a shortest path between their positions, because of the limit of the number of steps. In \probR, we have no such a limitation and Facilitator can use other strategies. However, we can modify the construction of the graph to make Facilitator behave exactly in the same way as in the above proof or loose immediately. 

\begin{figure}[ht]
\centering
\scalebox{0.7}{
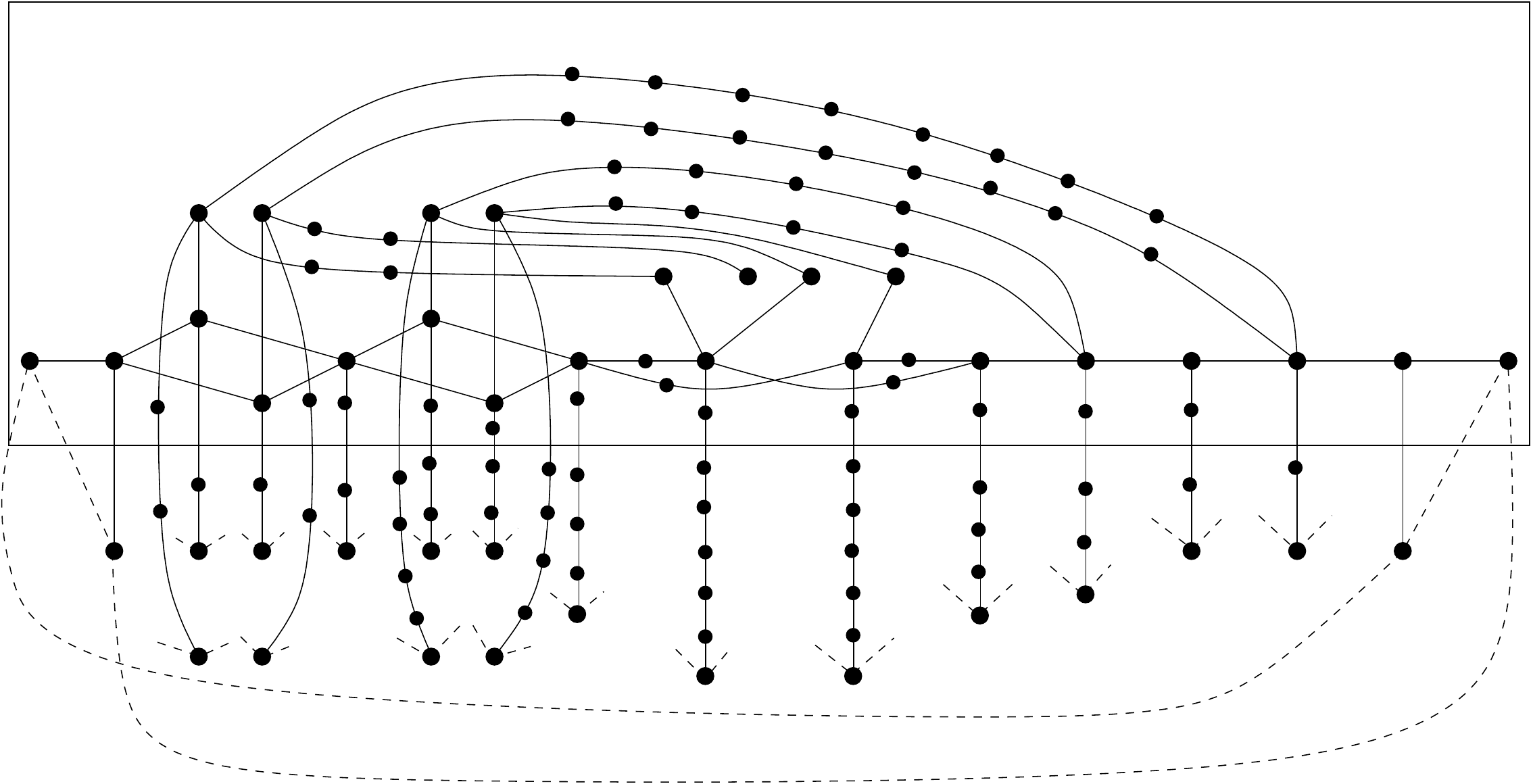}
\caption{The construction of $G'$ for $G$ shown in Figure~\ref{fig:ps-hardness}.}
\label{fig:ps-hardness-two}
\end{figure}  

We construct the graph $G'$ starting from $G$ as follows (see Figure~\ref{fig:ps-hardness-two}).
\begin{itemize}
\item Construct a copy of $G$.
\item For $i\in\{0,\ldots,n\}$, construct a vertex $u_i$, make it adjacent to $s$ and $t$, and join it with $u_i$ by a path $L_i$ of length $2i+1$.
\item For $i\in\{1,\ldots,n\}$, 
\begin{itemize}
\item construct vertices $a_i,\overline{a}_i,a_i',\overline{a}_i'$ and make them adjacent to $s$ and $t$, 
\item join $a_i$ with $x_{2i-1}$ by a path $S_i$ and join $\overline{a}_i$ with $\overline{x}_{2i-1}$ by a path $\overline{S}_i$ of length $2i+1$, 
\item join $a_i'$ with $x_{2i-1}'$ by a path $S_i'$ and join $\overline{a}_i'$ with $\overline{x}_{2i-1}'$ by a path $\overline{S}_i'$ of length $2i$. 
\end{itemize}
\item For $i\in\{0,\ldots,2n\}$, construct a vertex $v_i'$, make it adjacent to $s$ and $t$, and join it with $v_i$ by a path $L_i'$ of length $i+1$.
\item For $j\in\{1,\ldots,m\}$, construct a vertex $c_j'$, make it adjacent to $s$ and $t$, and join it with $c_j$ by a path $F_j$ of length $2n+3$. 
\end{itemize}
Let $Y=\{u_0,\ldots,u_n\}\cup\big(\bigcup_{i=1}^n\{a_i,\overline{a}_i,a_i',\overline{a}_i'\}\big)\cup\{v_0',\ldots,v_{2n}'\}\cup\{c_1',\ldots,c_m'\}$.
Then we define $k'=k+|Y|=9n+m+4$.

 We claim that $\varphi$ evaluates \true if and only if \divider with $k'$ agents has a winning strategy in Rendezvous Game with Adversaries.
 
 Assume that $\varphi=\true$. We describe a winning strategy for \divider. The $k=2n+2$ agents $D_1,\ldots,D_k$ are initially placed exactly as in the proof for \probRT. The remaining $|Y|$ agents are placed in the vertices of the set $Y$; we call these agents \emph{auxiliary}. The agents  $D_1,\ldots,D_k$ are using essentially the same strategy   as in the proof for $\probRT$ (we call this strategy \emph{old}). We use the same notation $X_i$ to denote the set of vertices occupied by these agents after the $i$-th step of the game.
 The auxiliary agents force Facilitator to move $R$ and $J$ in the same way as in the previous proof.
For  $i\geq 1$, we denote by $X_i'$ the set of vertices occupied by the agents of \divider after the $i$-th step of the game; $X_0'=X_0\cup Y$.
 
 We can assume that Facilitator moves ether $R$ or $J$ to an adjacent vertex by the first move. Suppose that $R$ is moved to $u_0$ and $J$ keeps the old position in $t$. Then \divider moves $D_{2n+1}$ from $z$ to $s$ and the agent from $v_0'$ is moved to $v_0$. Observe that $R$ and $J$ are now in distinct connected components of $G-X_1$ and \divider wins by the trivial strategy, that is, by keeping all the agents in their current position.  Similarly, if $J$ is moved to $v_0$ and $R$ remains in $s$, then 
 \divider moves $D_{2n+1}$ to $t$ and the agent from $u_0'$ is moved to $u_0$. Again, $X_1$ separates $R$ and $J$, that is, \divider wins. Assume that both $R$ and $J$ are moved in the first step of the game. Then \divider responds by moving $D_1,\ldots,D_{2n+2}$ using the old strategy. The auxiliary agents are moved to adjacent vertices along the paths $L_i$ for $i\in\{0,\ldots,n\}$, $S_i,\overline{S}_i,S_i',\overline{S}_i'$ for $i\in\{1,\ldots,n\}$, $L_i'$ for $u\in\{0,\ldots,2n\}$ and $F_j$ for $j\in\{1,\ldots,m\}$. By the subsequent moves, these agents are moved further along these paths until they reach the end-vertices. If an auxiliary agent is unable to enter a vertex, because it is occupied by an agent of Facilitator, \divider waits until the vertex get vacated and then moves the agent there.
 
 Assume inductively for $i=0,\ldots,n$ that after the $2i+1$-th step of the game, $R$ is in $u_i$, $J$ is in $v_{2i}$ and the agents $D_1,\ldots,D_{k}$ are occupying the positions according to the old strategy. 
 Notice if $i=0$, then $s$ and $t$ are occupied by $D_{2n+1}$ and $D_{2n+2}$. If $i\geq 1$, then the vertices $x_{2i-1}',\overline{x}_{2i-1}'$ and the vertex $v_{2i-1}$ are occupied by  auxiliary agents of \divider. Moreover, the vertices that are adjacent to $u_i$ in $L_i$ and to $v_{2i}$ in $L_{2i}'$ are also occupied by  auxiliary agents. 
 This means that neither $R$ or $J$ can move ``backward'' or use $L_{i}$ or $L_{2i}'$.
 
 Suppose that $i<n$. Notice that the vertices of $S_{i+1}'$ and  $\overline{S}_{i+1}'$ that are adjacent to $x_{2i+1}'$ and $\overline{x}_{2i+1}'$  are occupied by auxiliary agents. If Facilitator does not move $R$ to an adjacent vertex, i.e., either to  $x_{2i+1}'$ or $\overline{x}_{2i+1}'$, then \divider moves the agents to  $x_{2i+1}'$ and $\overline{x}_{2i+1}'$ and wins by the trivial strategy.  Similarly, the vertex adjacent to $v_{2i+1}$ in $L_{2i+1}'$ is occupied by an auxiliary agent. Hence, if $J$ is not moved, this agents enters $v_{2i+1}$ and  $J$ gets separated from $R$. 
 
 Assume that $R$ and $J$ are moved to adjacent vertices. \divider responds using the old strategy. Assume that $R$ is moved to $x_{2i+1}'$ as the other case is symmetric. Recall that according to the old strategy $D_{2i+1}$ is moved to $x_{2i+1}$. Notice also that $u_{i}$ and $v_{2i}$ are occupied by auxiliary agents. Moreover, the neighbors of 
 $x_{2i+1}$, $\overline{x}_{2i+1}$, $x_{2i+1}'$, $\overline{x}_{2i+1}'$, $u_{i+1}$,   $v_{2i+1}$ and $v_{2i+2}$
 in $S_{i+1}$, $\overline{S}_{i+1}$, $S_{i+1}'$, $\overline{S}_{i+1}'$, $L_{i+1}$ and $L_{2i+1}$, respectively, are  occupied by auxiliary agents. If $R$ is not moved, then an agent enters  $u_{i+2}$ and $R$ gets separated from $J$. If $J$ is not moved, then agents enter  $x_{2i+1}$, $\overline{x}_{2i+1}$ and  $v_{2i+2}$. Again, $R$ and $J$ are in distinct components of $G-X_{2i+2}$. Suppose that $J$ is moved to one of the neighbors of $v_{2i+1}$ in $Q_{2i+1}$ or $\overline{Q}_{2i+1}$. \divider responds by moving agents to $v_{2i+1}$,  $x_{2i+1}$ and $\overline{x}_{2i+1}$ and wins. We conclude that both $R$ and $J$ should be moved ``forward'' to $u_{i+1}$ and $v_{2i}$. Then \divider responds using the old strategy.   
 
 Using these arguments, we obtain after $2n+1$ steps of the game either \divider already separated $R$ and $J$ and won or 
  the following configuration is achieved:
\begin{itemize}
\item $R$ is in $u_n$ and $J$ is in $v_{2n}$ and the vertices of the sets $N_{G'}(u_n)\setminus\{w_1,\ldots, w_m\}$ and  $N_{G'}(v_{2n})\setminus\{w_1',\ldots, w_m'\}$ are occupied by auxiliary agents.
\item For every $j\in\{1,\ldots,m\}$, the vertex at distance two from $c_i$ in $F_i$ is occupied by  auxiliary agents. 
\item For $i\in\{1,\ldots,n\}$, $D_i$ is either  in $x_{i}''$ or in $\overline{x}_{i}''$, and for the corresponding choice of the values of  the Boolean variables $x_1,\ldots,x_{2n}$,
$\psi=C_1\wedge\ldots \wedge C_m=\true$.
\end{itemize}  
If Facilitator  moves neither  $R$ nor $J$ to adjacent vertices, \divider moves auxiliary agents in two steps to  $c_1,\ldots,c_m$ and separates $R$ and $J$. Assume that $R$ is moved from $u_n$ to $w_j$. Then the auxiliary agent that is in the vertex adjacent to $u_n$ in the path $L_n$ is moved to $u_n$ and one of the agents $D_1,\ldots,D_{2n}$ is moved to $c_j$. Recall that such an agent exists, because $C_j=\true$. Then $R$ is separated from $J$. Similarly, if $J$ is moved to $w_j'$, then an auxiliary agent is moved to $v_{2n}$ and one of the agents $D_1,\ldots,D_{2n}$ is moved to $c_j$. Then \divider wins.

This completes the proof that if $\varphi$ evaluates $\true$, then  \divider with $k'$ agents has a winning strategy.

To show that  if $\varphi$ evaluates \false, then Facilitator has a winning strategy, we use the same arguments as in the analogous proof for \probRT.
 If there is a vertex of  $ N_G(s)\cap N_G(t)$ that is not occupied by the agents of \divider, Facilitator wins in one step by moving $R$ and $J$ to this vertex.
 Assume that all the vertices of $ N_G(s)\cap N_G(t)$ are occupied by the agents of \divider in the beginning of the game. In particular, we have that every vertex of $Y$ is occupied by one agent. Now Facilitator uses the same strategy as in the proof for \probRT. To see that this is a winning strategy, it is sufficient to observe that the agents of \divider that are placed in the vertices of $Y$ are out of game and we can ignore them in the analysis of the strategy of Facilitator. 
 
 We obtain that $\varphi$ evaluates \true if and only if \divider with $k'$ agents has a winning strategy in Rendezvous Game with Adversaries. Therefore, \probR is \classPSPACE-hard.
 \end{proof}
 
 \section{\probRT for graph of bounded neighborhood diversity}\label{sec:diversity}
 In this section, we show that \probRT is \classFPT when parameterized by $\tau$ and the neighborhood  diversity of the input graph.
 
 The notion of neighborhood diversity was introduced by Lampis in~\cite{Lampis12}. It is convenient for us to define this notion in terms of modules. Let $G$ be a graph. A set of vertices $U\subseteq V(G)$ is a \emph{module} if for every $v\in V(G)\setminus U$, either $N_G(v)\cap U=\emptyset$ or $U\subseteq N_G(v)$. It is said that is a module $U$ is a \emph{clique module} if $U$ is a clique, and $U$ is an \emph{independent module} if $U$ is an independent set. We say that a partition $\{U_1,\ldots,U_\ell\}$  
 of $V(G)$ into clique and independent modules is a \emph{neighborhood decomposition}.
 The \emph{neighborhood diversity} of a graph $G$ is the minimum $\ell$ such that $G$ has a neighborhood decomposition with $\ell$ modules;
 we use $\nd(G)$ to denote the neighborhood diversity of $G$. 
 The value of $\nd(G)$ and the corresponding partition of $V(G)$ into clique and independent modules can be computed in polynomial (linear) time~\cite{Lampis12}. 
 Given a neighborhood decomposition $\mathcal{U}=\{U_1,\ldots,U_\ell\}$, we define the \emph{quotient} graph $\mathcal{G}$ as the graph with the vertex set $\{1,\ldots,\ell\}$ such that $i$ is adjacent to $j$ for distinct $i,j\in\{1,\ldots,\ell\}$ if and only if $u\in U_i$ is adjacent to $v\in U_j$ in $G$. For a vertex $v\in V(G)$, $\id(v)=i$ if $v\in U_i$.
 For a multiset of vertices $X=\{x_1,\ldots,x_r\}$, $\id(X)$ denotes the multiset of indices $\{\id(x_1),\ldots,\id(x_r)\}$.
 
Let  $\mathcal{U}=\{U_1,\ldots,U_\ell\}$ be a neighborhood decomposition of $G$. Notice that every bijective mapping $\varphi\colon V(G)\rightarrow V(G)$ such that 
$\varphi(U_i)=U_i$ for $i\in \{1,\ldots,\ell\}$ is an automorphism of $G$. We say that $\varphi$ is an automorphism that \emph{agrees with $\mathcal{U}$}.

For an automorphism $\varphi$, we extend it on multisets  of vertices in the natural way. Namely, if $X=\{x_1,\ldots,x_r\}$ is a multiset of vertices of $G$, 
$\varphi(X)=\{\varphi(x_1),\ldots,\varphi(x_r)\}$. Similarly, for a pair $(X,Y)$ of multisets without common elements, $\varphi(X,Y)=(\varphi(X),\varphi(Y))$.

%
%
%
%
 
 Let $G$ be a connected graph and let $\mathcal{U}=\{U_1,\ldots,U_\ell\}$ be a neighborhood decomposition of $G$ such that $\ell=\nd(G)$. Suppose that $s,t\in V(G)$ such that $s$ and $t$ are distinct modules of $\mathcal{U}$.  We consider our Rendezvous Game with Adversaries on $G$ in $\tau$ steps. 
 
 Consider a strategy of \divider with $k$ agents, that is, a family of functions $d_i\colon \mathcal{P}_G^k\rightarrow \mathcal{D}_G^k$ for $i\in\{0,\ldots,\tau-1\}$, where 
  $d_i\colon \mathcal{P}_G^k\rightarrow \mathcal{D}_G^k$ for $i\in\{0,\ldots,\tau-1\}$ such that for $i\in\{1,\ldots,\tau-1\}$. Recall that
$d_i$ maps $(F,D)\in \mathcal{P}_G^k$ to $D'\in \mathcal{F}_G^k$, where  $D$ and $D'$ are adjacent and $D'$ is compatible with $F$, and $d_0$ maps $(\{s,t\},\emptyset)$ to $D'\in\mathcal{D}_G^k$ compatible with $\{s,t\}$.

 Recall that a  strategy of \divider can be  represented as a rooted tree $\mathcal{T}_G^k(\tau)$ 
 of height $\tau$. Each node $v\in V(\mathcal{T}_G^k(\tau))$ is associated with a position $P_v\in \mathcal{P}_G^k$, and 
 \begin{itemize}
\item $P_r=(\{s,t\},d_0(\{s,t\},\emptyset))$ is associated with the root $r$ of $\mathcal{T}_G^k(\tau)$,
\item for every node $v\in V(\mathcal{T}_G^k(\tau))$ with $P_v=(F,D)$ at 
at distance $i\leq\tau-1$ from the root, there is a child $u$ of $v$ for every $(F',D')\in \mathcal{P}_G^k$ such that 
(i) $F'$ is adjacent to $F$ and compatible with $D$, and (ii) $D'=d_i(F',D)$, and $u$ is associated  with  $P_u=(F',D')$. 
\end{itemize} 
From now, we consider such a representation.

By Observation~\ref{obs:win}, 
$\mathcal{T}_G^k(\tau)$ is a winning strategy for \divider if and only if  $F$ is a set of two distinct vertices for every $P_v=(F,D)$ for  $v\in V(\mathcal{T}_G^k(\tau))$.
By our assumption that $s$ and $t$ are in distinct modules, we can refine the claim.

\begin{observation}\label{obs:win-ref}
The tree $\mathcal{T}_G^k(\tau)$ is a winning strategy for \divider if and only if  for every $v\in V(\mathcal{T}_G^k(\tau))$ with  $P_v=(F,D)$ for  $v\in V(\mathcal{T}_G^k(\tau))$, $F$ contains at most one vertex of every module $U_i$ for $i\in\{1,\ldots,\ell\}$.
\end{observation}

\begin{proof}
To see the observation, it is sufficient to note that if both agents of \divider are moved to the same module or if one of the agents is in a module $U_i$ and the other is moved to $U_i$, then Facilitator can move the agents into the same vertex.
\end{proof}

  Let $P=(F,D)$ and $P'=(F',D')$ be positions in Rendezvous Game with Adversaries on $G$. We say that $P$ and $P'$ are isomorphic, if there an automorphism $\varphi$ of $G$ such that $P'=\varphi(P)$. We also say that $P$ and $P'$ are \emph{isomorphic with respect to $\varphi$} for such an automorphism $\varphi$.  
 We use the following straightforward observation about positions in the game.

 \begin{observation}\label{obs:isom-pos}
 Let $P$ and $P'$ be isomorphic positions in Rendezvous Game with Adversaries on $G$. Then \divider can win in at most $r$ steps if the game starts from $P$ if and only if \divider can win in $r$ steps if the game starts from $P'$.  
  \end{observation}
 
We say that  $\mathcal{T}_G^k(\tau)$ is a \emph{uniform} strategy if for every node $v$ with $P_v=(F,D)$ and each of its two children $u_1$ and $u_2$  with $P_{u_1}=(F_1,D_1)$ and $P_{u_2}=(F_2,D_2)$, the following holds: if $(F_1,D)$ and $(F_2,D)$ are isomorphic with  respect to some automorphism $\varphi$ of $G$ that agrees with $\mathcal{U}$, then $P_{u_1}$ and $P_{u_2}$ are isomorphic with resect to some automorphism $\psi$ of $G$ that agrees with $\mathcal{U}$.
Informally, if possible moves of the agents of Facilitator to $F_1$ and $F_2$ are the same with respect to moving them to the same modules, then the response of \divider is also the same (up to an automorphism that agrees with $\mathcal{U}$). Observation~\ref{obs:isom-pos} immediately implies the following.

 \begin{observation}\label{obs:uiniform} 
If \divider has a winning strategy  in Rendezvous Game with Adversaries on $G$, then \divider has a uniform winning strategy.  
\end{observation}

From now on, we assume that  $\mathcal{T}_G^k(\tau)$ is uniform.

Let $u_1$ and $u_2$ be distinct children of a node $v$ of $\mathcal{T}_G^k(\tau)$. We say that $u_1$ and $u_2$ are \emph{equivalent} if $P_{u_1}$ and $P_{u_2}$ are isomorphic with respect to some   automorphism $\varphi$ of $G$ that agrees with $\mathcal{U}$. We also say that two subtrees $T_1$ and $T_2$ rooted in $u_1$ and $u_2$ are \emph{equivalent} if $u_1$ and $u_2$ are equivalent.  It is straightforward  to see that the introduced relation is indeed an equivalence relation.  
Observe that, because the strategy is uniform, for 
$P_{u_1}=(F_1,D_1)$ and $P_{u_2}=(F_2,D_2)$, $u_1$ and $u_2$ are equivalent  if and only if $|F_1\cap U_i|=|F_2\cap U_i|$ for all $i\in\{1,\ldots,\ell\}$.

Since  $\mathcal{T}_G^k(\tau)$ is uniform, to represent the strategy, it is sufficient to keep one representative from each class of equivalent children. Given   
$\mathcal{T}_G^k(\tau)$, we construct the \emph{reduced strategy} $\hat{\mathcal{T}}_G^k(\tau)$ obtained by the following operation applied top-down staring from the root:
for a node $v$ and a class of equivalent subtrees rooted in the children of $v$, delete all the elements of the class except one.  Observe that given a reduced strategy $\hat{\mathcal{T}}_G^k(\tau)$, we can reconstruct   $\mathcal{T}_G^k(\tau)$.  Notice also that by Observations~\ref{obs:win-ref} and \ref{obs:isom-pos}, the strategy is a winning strategy for \divider if and only if  for every $v\in V(\hat{\mathcal{T}}_G^k(\tau))$ with  $P_v=(F,D)$ for  $v\in V(\mathcal{T}_G^k(\tau))$, $F$ contains at most one vertex of every module $U_i$ for $i\in\{1,\ldots,\ell\}$. 

Now we construct the tree that represents all possible moves of Facilitator in $\tau$ steps between the modules without the agents of \divider. We define the rooted tree  $\mathcal{T}_G^*(\tau)$ 
 of height $\tau$ with each node $v$ associated with a pair $X_v=\{p,q\}$ of not necessarily distinct elements of $\{1,\ldots,\ell\}$
 such that
 \begin{itemize}
\item $X_r=\{s,t\}$ is associated with the root $r$ of $\mathcal{T}_G^*(\tau)$,
\item for every node $v\in V(\mathcal{T}_G^*(\tau))$ with $X_v=\{p,q\}$ at distance at most $\tau-1$ from the root, there is a child $u$ of $v$ with $X_v=\{p',q'\}$ for every $\{p',q'\}$ adjacent to $\{p,q\}$ in the quotient graph $\mathcal{G}$.
\end{itemize} 
Observe that $\mathcal{T}_G^*(\tau)$ has at most $\binom{\ell+1}{2}^{\tau+1}$ nodes.

The tree $\hat{\mathcal{T}}_G^k(\tau)$ can be seen as a subtree of $\mathcal{T}_G^*(\tau)$. Formally, we define an injective mapping $\alpha\colon V(\hat{\mathcal{T}}_G^k(\tau))\rightarrow V(\hat{\mathcal{T}}_G^k(\tau))$ inductively top-down:
\begin{itemize}
\item for the root $r$ of $\hat{\mathcal{T}}_G^k(\tau)$, $\alpha(r)$  is the root of $\mathcal{T}_G^*(\tau)$,
\item if $\alpha(v)=u$ for $v\in V(\hat{\mathcal{T}}_G^k(\tau))$, then every child $v'$ of $v$ in $\hat{\mathcal{T}}_G^k(\tau)$ with $P_{v'}=(F,D)$ is mapped to the child $u'$ of $u$ in $\mathcal{T}_G^*(\tau)$
with $X_{u'}=\id(F)$.
\end{itemize} 
In particularly,  for every $v\in V(\hat{\mathcal{T}}_G^k(\tau))$ with  $P_{v}=(F,D)$, $X_{\alpha(v)}=\id(F)$. We say that the subtree of $\mathcal{T}_G^*(\tau)$ induced by 
$\alpha(V(\hat{\mathcal{T}}_G^k(\tau)))$ is a \emph{projection} of  $\hat{\mathcal{T}}_G^k(\tau)$ to $\mathcal{T}_G^*(\tau)$.
We use the following property of projections that immediately follows from the definition. 

\begin{observation}\label{obs:proj}
Let $u$ be a non-leaf node of $\mathcal{T}_G^*(\tau)$ with $X_{u}=\{p,q\}$. 
Let also  $u=\alpha(v)$ for some $v\in V(\hat{\mathcal{T}}_G^k(\tau))$ with $P_v=(F,D)$ and let 
$I_v=\{i\mid i\in\{1,\ldots,\ell\}\text{ and }U_i\subseteq D\}$.
Then a child $u'$ of $u$ in $\mathcal{T}_G^*(\tau)$ with $X_{u'}=\{p,q\}$ is a child of $u$ in the projection  of $\hat{\mathcal{T}}_G^k(\tau)$
if and only if $\{p',q'\}\in\{\{i,j\}\mid i\in N_{\mathcal{G}}[p]\setminus I_v\text{ and }j\in N_{\mathcal{G}}[q]\setminus I_v\}$.
\end{observation}

Note, in particular, that each leaf of the projection  of $\hat{\mathcal{T}}_G^k(\tau)$ is a leaf of $\mathcal{T}_G^*(\tau)$.

In our algorithm for \probRT, we check whether \divider with $k$ agents has a winning strategy on $G$. For this, we consider  $\mathcal{T}_G^*(\tau)$ and guess the projection $\mathcal{T}$ of a hypothetical reduced winning strategy tree by trying all subtrees of $\mathcal{T}_G^*(\tau)$ using brute force. For each $\mathcal{T}$, we verify whether \divider indeed has a strategy corresponding to $\mathcal{T}$ by checking whether \divider can respond to the moves of Facilitator in such a way that \divider is able to ensure that $\mathcal{T}$ has the required structure, according to Observation~\ref{obs:proj}.    

Checking whether \divider has a strategy corresponding to $\mathcal{T}$ 
is based on the results of 
 Lenstra~\cite{Lenstra83} (see also~\cite{Kannan87,FrankT87} for further improvements) about parameterized complexity of Integer Linear Programming. 
 The task of the \probILP problem is, given a $q \times p$ matrix $A$ over $\mathbb{Z}$ and a vector $b\in \mathbb{Z}^q$, decide whether there is a vector $x\in\mathbb{Z}^p$ such that $Ax\leq b$; we write $Ax\leq b$ to denote that for every $i\in\{1,\ldots,q\}$, the $i$-th element of the vector $Ax$ is at most the $i$-th element of $b$.
 Lenstra~\cite{Lenstra83} proved that \probILP is \classFPT when parameterized by $p$ and later this result was improved by Kannan~\cite{Kannan87}. Further, Frank and Tardos~\cite{FrankT87} proved that \probILP can be solved in polynomial space.  These results can be summarized in the following statement.
 
 \begin{proposition}[\cite{Lenstra83,Kannan87,FrankT87}]\label{prop:ipl} 
 \probILP can be solved in $\Oh(p^{2.5p+o(p)}\cdot L)$ time and polynomial in $L$ space, where $L$ is the number of bits in the input.
 \end{proposition} 
 
 \begin{theorem}\label{thm:nd}
 \probRT can be solved in $2^{\ell^{\Oh(\tau)}}\cdot n^{\Oh(1)}$ time on graphs of neighborhood diversity $\ell$. 
 \end{theorem}
 
\begin{proof}
Let $(G,s,t,k,\tau)$ be an instance of \probRT. If $s=t$ or $s$ and $t$ are adjacent, then the problem is trivial. Assume that $s$ and $t$ are distinct nonadjacent vertices of $G$. 
We compute a neighborhood decomposition $\mathcal{U}=\{U_1,\ldots,U_\ell\}$ of $G$ with $\ell=\nd(G)$. Recall that this can be done in polynomial time~\cite{Lampis12}. 
Denote by $n_i=|U_i|$ for $i\in\{1,\ldots,\ell\}$.

Suppose that $s$ and $t$ are in the same module $U_i$. Since $s$ and $t$ are distinct and not adjacent,   $U_i$ is an independent module. We have that $N_G(s)=N_G(t)$ and, therefore, $\lambda_G(s,t)=|N_G(s)\cap N_G(t)|$. Notice that Facilitator wins in one step if $k<|N_G(s)\cap N_G(t)|$ by moving $R$ and $J$ into a vertex of $N_G(s)\cap N_G(t)$ that is not occupied by an agent of \divider. We conclude that $d_G(s,t)=|N_G(s)\cap N_G(t)|$ and, therefore,  $(G,s,t,k,\tau)$ is a yes-instance if and only if $k<|N_G(s)\cap N_G(t)|$. From now, we assume that $s$ and $t$ are in distinct modules of $\mathcal{U}$. 

We construct the tree $\mathcal{T}_G^*(\tau)$ by brute force with the corresponding pairs $X_v=\{p,q\}$ for $v\in V(\mathcal{T}_G^*(\tau))$. Since  $\mathcal{T}_G^*(\tau)$ has at most $\binom{\ell+1}{2}^{\tau+1}$ nodes, the construction can be done in $\ell^{\Oh(\tau)}$ time. Denote by $r$ the root of $\mathcal{T}_G^*(\tau)$.

We consider all subtrees $\mathcal{T}$ of $\mathcal{T}_G^*(\tau)$ containing $r$ and rooted in this vertex, whose leaves are leaves of $\mathcal{T}_G^*(\tau)$.
Observe that the total number of such trees is at most $2^{|\mathcal{T}_G^*(\tau)|}=2^{\ell^{\Oh(\tau)}}$. For each $\mathcal{T}$, we check whether \divider has a winning strategy such that the projection of the corresponding reduced strategy is $\mathcal{T}$. If we find such a tree $\mathcal{T}$, we conclude that \divider wins in the game. Otherwise, we conclude that Facilitator wins.

Assume that $\mathcal{T}$ is given. 
If for some $v\in V(\mathcal{T})$, $X_v=\{p,p\}$ for some $p\in\{1,\ldots,\ell\}$, we discard the choice of $\mathcal{T}$, because $\mathcal{T}$ cannot be the projection of a winning strategy of \divider by Observation~\ref{obs:win-ref}. Suppose that for every $v\in V(\mathcal{T})$, $X_{v}=\{p,q\}$ with $p\neq q$.

The running time of our algorithm is going to be dominated by checking all the trees $\mathcal{T}$ and solving \probILP. Therefore, to simplify the arguments, for each node $v$ of $\mathcal{T}$ we guess the set  $I_v\subseteq \{1,\ldots,\ell\}$ such that the agents of \divider occupy all the vertices 
of the modules $U_i$ with $i\in I_v$  in the position of the game corresponding to $v$. As standard, we do it by brute force checking of all possible assignments of sets to the nodes of $\mathcal{T}$. Since the number of the assignments is at most $(2^\ell)^{|V(\mathcal{T})|}$, this can be done in  $2^{\ell^{\Oh(\tau)}}$ time. 

For each selection of $I_v$ for $v\in V(\mathcal{T})$, we check feasibility using Observation~\ref{obs:proj}. Namely, for each non-leaf vertex $v\in V(\mathcal{T})$, we consider its set $X_v=\{p,q\}$ and check 
whether the children $u$ of $v$ in  $\mathcal{T}_G^*(\tau)$ with $X_{u}\in\{\{i,j\}\mid i\in N_{\mathcal{G}}[p]\setminus I_v\text{ and }j\in N_{\mathcal{G}}[q]\setminus I_v\}$ are exactly the children of $v$ in $\mathcal{T}$. We discard the assignment if this is not the case, and we discard the current choice of $\mathcal{T}$ if we fail to find a feasible assignment of sets $I_v$. 

From now on, we assume that the assignment of sets $I_v$ for $v\in V(\mathcal{T})$ is given.

Our general idea is to express the question about existence of a winning strategy of \divider in terms of  \probILP. We start with introducing two families of variables 
$x_1^v,\ldots,x_\ell^v$ and $y_1^v,\ldots,y_\ell^v$ for each node $v$ of $\mathcal{T}$.  The intuition behind these variables is following. For every $i\in\{1,\ldots,\ell\}$, $x_i^v$ is the number of vertices of $U_i$ occupied by agents of \divider in the position of the game corresponding to the node $v$. It is more convenient for us to consider $x_i^v$ as the number 
of the agents of \divider that occupy distinct vertices of $U_i$; we call these agents \emph{blockers}.    
\divider may also have other agents in $U_i$ and $y_i^v$ is the number of  these agents and we call these agents \emph{dwellers}. It is also convenient to assume that blockers are active in the current step of the game, and dwellers are inactive and do not prevent $R$ or $J$ from entering the vertices occupied by them. By this convenience, we can allow, say, the situation $x_i^v=0$ and $y_i^v>0$, as we do not care where the dwellers are placed in the corresponding module. 

We impose the following constraints on these variables for every $v\in V(\mathcal{T})$:
\begin{gather}
\sum_{i=1}^\ell(x_i^v+x_i^v)=k,~ x_i^v\geq 0 \text{ and }y_i^v\geq 0\text{ for every }i\in\{1,\ldots,\ell\}, \label{eq:const-one}\\
x_i^v\leq n_i\text{ for every }i\in\{1,\ldots,\ell\}\setminus X_v,\label{eq:const-two}\\
x_i^v\leq n_i-1\text{ for every }i\in X_v,\label{eq:const-three}\\
y_i^v=0\text{ for }i\in\{1,\ldots,\ell\}\text{ if }n_i=1\text{ and }i\in X_v.  \label{eq:const-four}
\end{gather}
The necessity of constraints (\ref{eq:const-one}) and (\ref{eq:const-two}) is straightforward.  To see the reason behind (\ref{eq:const-three}) and (\ref{eq:const-four}), notice that if a vertex of $U_i$ is occupied by an agent of Facilitator, then at most $n_i$ blockers can be in $U_i$ and, moreover, if $n_i=1$, then no agent of \divider can be in $U_i$.

Next, we state the constraints coming from the choice of sets $I_v$. 
For every $v\in V(\mathcal{T})$,
\begin{equation}\label{eq:const-five}
x_i^v=n_i\text{ for every }i\in  I_v. 
\end{equation}

The variables $x_i^v$ and $y_i^v$ are used to express the positions of the players. However, we also have to express transitions between these positions, that is, the players should be able to make moves from the position corresponding to a node of $\mathcal{T}$ to the positions corresponding to its children. For this, we need additional variables. 
For every $v\in V(\mathcal{T})$ and every child $u$ of $v$ in $\mathcal{T}$, and every ordered pair $(i,j)$ of adjacent vertices of the quotient graph $\mathcal{G}$, we introduce four variables $a_{i,j}^{v,u},b_{i,j}^{v,u},c_{i,j}^{v,u},d_{i,j}^{v,u}$. The meaning of the variables is following. For the move of Facilitator from the position corresponding to $v$ to the position corresponding to $u$, \divider responds by moving $a_{i,j}^{v,u}$ blockers from $X_i$ to make them blockers in $X_j$, $b_{i,j}^{v,u}$ blockers from $X_i$ become dwellers in $X_{j}$,  $c_{i,j}^{v,u}$ dwellers from $X_i$ become blockers in $X_{j}$, and $d_{i,j}^{v,u}$ dwellers from $X_i$ become dwellers in $X_{j}$.
Notice that if $U_i$ is a clique module, then some dwellers can move to adjacent vertices to become blockers (it has no sense for \divider to make a blocker a dweller). For this, we introduce a variable $z_i^{v,u}$ for $i\in\{1,\ldots,\ell\}$.
The constraints for these variables are following.

For every non-leaf $v\in V(\mathcal{T})$ and every child $u$ of $v$ in $\mathcal{T}$, 
\begin{gather}
a_{i,j}^{v,u}\geq 0,~b_{i,j}^{v,u}\geq 0,~c_{i,j}^{v,u}\geq 0,~d_{i,j}^{v,u}\geq 0 \text{ for each ordered pair }(i,j)\text{ of adjacent vertices of }\mathcal{G},\label{eq:const-six}\\
z_i^{v,u}\geq 0\text{ and }z_i^{v,u}=0\text{ if }U_i\text{ is an independent module}\text{ for every }i\in\{1,\ldots,\ell\},\label{eq:const-seven}\\
\sum_{j\in N_{\mathcal{G}}(i)}(a_{i,j}^{v,u}+b_{i,j}^{v,u})\leq x_i^{v}\text{  and  }\sum_{j\in N_{\mathcal{G}}(i)}(c_{i,j}^{v,u}+d_{i,j}^{v,u})+z_i^{v,u}\leq y_i^v\text{ for every }i\in\{1,\ldots,\ell\},\label{eq:const-eight}\\
x_i^u=x_i^v-\sum_{j\in N_{\mathcal{G}}(i)}(a_{i,j}^{v,u}+b_{i,j}^{v,u})+\sum_{j\in N_{\mathcal{G}}(i)}(a_{j,i}^{v,u}+c_{j,i}^{v,u})+z_i^{v,u}\text{ for every }i\in\{1,\ldots,\ell\},\label{eq:const-nine}\\
y_i^u=y_i^v-\sum_{j\in N_{\mathcal{G}}(i)}(c_{i,j}^{v,u}+d_{i,j}^{v,u})-z_i^{v,u}+\sum_{j\in N_{\mathcal{G}}(i)}(b_{j,i}^{v,u}+d_{j,i}^{v,u})\text{ for every }i\in\{1,\ldots,\ell\}.\label{eq:const-ten}
\end{gather}
Constraints (\ref{eq:const-six}) and (\ref{eq:const-seven}) are straightforward. Constraint (\ref{eq:const-eight}) encodes that the number of blockers that leave a module $U_i$ is upper bounded by the number of blockers in $U_i$ and, symmetrically, the number of dwellers that leave a module $U_i$ or become blockers in the block is at most the number of dwellers in $U_i$. Finally, (\ref{eq:const-nine}) and (\ref{eq:const-ten}) express that the movements of agents of \divider from the position associated 
with $v$ lead to the position corresponding to $u$.

We have $2|V(\mathcal{T})|\ell$ variables $x_i^v,y_j^v$, at most $8|E(\mathcal{T})|\binom{\ell}{2}$ variables   $a_{i,j}^{v,u},b_{i,j}^{v,u},c_{i,j}^{v,u},d_{i,j}^{v,u}$, and  $|E(\mathcal{T})|\ell$ variables $z_i^{u,v}$, that is, $\ell^{\Oh(\tau)}$ variables.  We defined $5|V(\mathcal{T})|\ell$ constraints (\ref{eq:const-one})--(\ref{eq:const-four}), at most $|V(\mathcal{T})|\ell$ constraints (\ref{eq:const-five}), and $|E(\mathcal{T})|(8\binom{\ell}{2}+5\ell)$ constraints (\ref{eq:const-six})--(\ref{eq:const-ten}). Hence, in total, we have $\ell^{\Oh(\tau)}$ constraints.  Denote the obtained system of integer linear inequalities by ($*$). 
Observe that the coefficients in ($*$) are upper bounded by $n$. Therefore, the bit-size of ($*$) is $\ell^{\Oh(\tau)}\cdot \log n$. We solve ($*$) in $2^{\ell^{\Oh(\tau)}}\cdot \log n$ time by Proposition~\ref{prop:ipl}. 

We claim that ($*$) is feasible, that is, has an integer solution if and only if \divider has a winning strategy such that the projection of the reduced strategy on $\mathcal{T}_G^*(\tau)$ is $\mathcal{T}$.

Suppose that \divider with $k$ agents has a uniform winning strategy $\mathcal{T}_G^k(\tau)$ such that the projection of the reduced strategy  $\hat{\mathcal{T}}$ on $\mathcal{T}_G^*(\tau)$ is $\mathcal{T}$. For every two distinct modules $U_i$ and $U_j$, ether every vertex of $U_i$ is adjacent to every vertex of $U_j$ or the vertices of the modules are nonadjacent. This allows us to make some additional assumptions about $\mathcal{T}_G^k(\tau)$. Namely, we can assume that on each step the agents of \divider are divided into blockers and dwellers, and then we can assume that a blocker (dweller, respectively) agent can become a dweller (blocker, respectively) only if the agent is moved to an adjacent vertex. Also we can assume that a blocker is never moved to an adjacent vertex of the same clique module to become a dweller.  
Then we define the values of all the variables according to the description given in the construction of ($*$) following the  
reduced strategy $\hat{\mathcal{T}}_G^k(\tau)$. Then the construction of the constraints of ($*$) immediately imply that these values of the variables provide a solution of ($*$).

For the opposite direction, given a solution of ($*$), we construct the strategy  $\mathcal{T}_G^k(\tau)$. Initially, we place the agents of \divider on $G$ according to the values of $x_1^r,\ldots,x_\ell^r$ and $y_1^r,\ldots,y_\ell^r$ for the root $r$. For each $i\in\{1,\ldots,\ell\}$, we place $x_i^r$ agents (blockers) into distinct vertices of $U_i$ unoccupied by the agents of Facilitator. Then we put $y_i^r$ dwellers into $U_i$; as we pointed above, it is convenient to assume that these agents are inactive and we can place them arbitrarily. Assume inductively that we constructed a node $v$ of the future  $\mathcal{T}_G^k(\tau)$ with $P_v=(F,D)$ that corresponds to the node $v'$ of $\mathcal{T}$, that is,
$\id(F)=X_{v'}$ and for each $i\in\{1,\ldots,\ell\}$, \divider has exactly $x_i^{v'}$ blockers in $U_i$ that occupy distinct vertices, and also $y_i^{v'}$ dwellers are in $U_i$.
Assume that $F'\in\mathcal{F}$ is compatible with $D$ and adjacent to $F$. Then because of constraints (\ref{eq:const-five}), there is a child $u'$ of $v'$ in $\mathcal{T}$ with 
$\id(F')=X_{u'}$. Then \divider responds to moving the agents of Facilitator from $F$ to $F'$ by moving his agents according to the values  
$a_{i,j}^{v',u'},b_{i,j}^{v',u'},c_{i,j}^{v',u'},d_{i,j}^{v',u'}$ for the ordered pairs $(i,j)$ of adjacent vertices of $\mathcal{G}$ and according to the values of $z_i^{v',u'}$ for $i\in\{1,\ldots,\ell\}$. For the obtained node $u$ of   $\mathcal{T}_G^k(\tau)$ with $P_u=(F',D')$, we have that the position corresponds to the configuration defined by the variables   $x_1^{u'},\ldots,x_\ell^{u'}$ and $y_1^{u'},\ldots,y_\ell^{u'}$. These inductive arguments imply that the constructed strategy is a uniform winning strategy for \divider and the projection of the reduced strategy is $\mathcal{T}$.

This completes the construction of the algorithm. To evaluate the running time, observe that we consider $2^{\ell^{\Oh(\tau)}}$ trees $\mathcal{T}$, and for each $\mathcal{T}$, we consider $2^{\ell^{\Oh(\tau)}}$ assignments of sets $I_v$ for the nodes. Then for each tree $\mathcal{T}$ given together with the assignments of sets $I_v$ for $v\in V(\mathcal{T})$, we construct a solve the system ($*$) in time $2^{\ell^{\Oh(\tau)}}\cdot \log n$. Taking into account the preliminary steps where we consider special cases of $s$ and $t$ and construct the neighborhood decomposition $\mathcal{U}$, the total running time is  $2^{\ell^{\Oh(\tau)}}\cdot n^{\Oh(1)}$.
\end{proof}
 
\section{Conclusion}\label{sec:concl}
We initiated the study of the Rendezvous Game with Adversaries on graphs.  We proved that in several cases, the \dinumber number $d_G(s,t)$, the minimum number of agents needed for \divider to win against Facilitator, could be equal to the minimum size $\lambda_G(s,t)$ of an $(s,t)$-separator in $G$.  In particular,   this equality holds on for $P_5$-free and chordal graphs. In general,   the difference $\lambda_G(s,t)-d_G(s,t)$ could be arbitrary large.  Are there other natural graph classes with this property? Is it is possible to characterize hereditary graph classes for which the equality holds?

Further, we investigated the computational complexity of \probR and \probRT. Both problems can be solved it $n^{\Oh(k)}$ time. However, they are \classCoW{2}-hard when parameterized by $k$ and cannot be solved in  $n^{o(k)}$ time unless $\classFPT=\classW{1}$. In fact,  $\tau$-\probRT is \classCoW{2}-hard for every $\tau\geq 2$. We also proved that \probR and \probRT are \classPSPACE-hard. We conjecture that these twoproblems are   {\sf EXPTIME}-complete.  

Finally, we initiated  the study of the complexity of  \probR and \probRT under structural parameterization of the input graphs. We proved that \probRT is \classFPT when parameterized by the neighborhood diversity of the input graph and $\tau$.  Can this result be generalized for the parameterization by \emph{modular width} (see, e.g.,~\cite{GajarskyLO13} for the definition and the discussion of this parameterization) and $\tau$?
Is \probRT \classFPT when parameterized by the neighborhood diversity only?  The same question is open for \probR. We believe that this problem is interesting even   
for the more restrictive parameterization by the vertex cover number. Another question is about  \probR and \probRT parameterized by the treewidth of a graph? 
Are the problems \classFPT or \classXP for this parameterization? Notice that if the initial positions $s$ and $t$ are not is the same bag of a tree decomposition of width $w$, then 
the upper bound for the \dinumber number by $\lambda_G(s,t)$ together with Theorem~\ref{thm:backtrack}
imply  that the problems can be solved it time $n^{\Oh(w)}$.  Can the problems be solved in this time if $s$ and $t$ are in the same bag?


\end{document}